\documentclass{CSML}

\def\dOi{11(4:4)2015}
\lmcsheading%
{\dOi}
{1--33}
{}
{}
{Jan.~13, 2014}
{Nov.~\phantom04, 2015}
{}

\ACMCCS{[{\bf Theory of computation}]: Logic---Logic and
  verification\,|\,Verification by model checking;  
   [{\bf Computing methodologies}]: Artificial
  intelligence---Knowledge representation and reasoning---Temporal
  reasoning; Artificial
  intelligence---Reasoning about belief and knowledge\,/\,Causal
  reasoning and diagnostics; [{\bf Hardware}]: Robustness---Fault
  tolerance\,/\,Hardware reliability} 

\subjclass{ I.2.4 [ARTIFICIAL INTELLIGENCE] Knowledge Representation
  Formalisms and Methods---Temporal Logic;B.8.1 [PERFORMANCE AND
  RELIABILITY] Reliability, Testing, and Fault-Tolerance}

\usepackage{hyperref}
\usepackage{subfig}

\usepackage{amsmath}
\usepackage{algpseudocode}
\usepackage{multirow}
\usepackage{timing}
\usepackage{colortbl}
\usepackage{xspace}
\usepackage{marginnote}
\usepackage{algpseudocode}

\usepackage{tikz}
\usetikzlibrary{positioning, automata, arrows, calc, shapes, patterns}

\newcommand{\props}{\ensuremath{{\mathcal{P}}}}
\newcommand{\diags}{\ensuremath{{\mathcal{D}}}}
\newcommand{\alarms}{\ensuremath{{\mathcal{A}}}}
\newcommand{\spec}{\ensuremath{{\alarms}}}
\newcommand{\system}[1][]{\ensuremath{\mktuple{V^{#1}, E^{#1}, I^{#1}, \mathcal{T}^{#1}}}}
\newcommand{\posystem}[1][]{\ensuremath{\mktuple{V^{#1}, E^{#1},
      I^{#1}, \mathcal{T}^{#1}, E^{#1}_0}}}

\newcommand{\Obs}{\ensuremath{{obs}}}
\newcommand{\ObsEquiv}{\ensuremath{\mbox{\textit{ObsEq}}}}
\newcommand{\ObsPoint}{\ensuremath{\mbox{\textit{ObsPoint}}}}
\newcommand{\strong}[1]{\ensuremath{{}_\llcorner\hspace*{-0.5mm}{#1}\hspace*{-0.5mm}{}_\lrcorner}}
\newcommand{\trigger}[1]{\ensuremath{\strong{#1}}}

\newcommand{\mktuple}[1]{\ensuremath{\langle{#1}\rangle}}
\newcommand{\ignore}[1]{}

\newcommand{\todo}[1]{}

\newcommand{\klone}{\ensuremath{KL_1}}

\newcommand{\kasl}{ASL$_K$\xspace}

\newcommand\edel[3]{\textsc{ExactDel(\ensuremath{{#1},{#2},{#3}})}}
\newcommand\bdel[3]{\textsc{BoundDel(\ensuremath{{#1},{#2},{#3}})}}
\newcommand\fdel[2]{\textsc{FiniteDel(\ensuremath{{#1},{#2}})}}

\newcommand\kedel[5]{\textsc{ExactDel$_K$(\ensuremath{{#1},{#2},{#3},{#4},{#5}})}}
\newcommand\kbdel[5]{\textsc{BoundDel$_K$(\ensuremath{{#1},{#2},{#3},{#4},{#5}})}}
\newcommand\kfdel[4]{\textsc{FiniteDel$_K$(\ensuremath{{#1},{#2},{#3},{#4}})}}

\newcommand\singlewave[1]{
\begin{tikzpicture}[draw=black, yscale=.4,xscale=0.2]
 \tikzstyle{time}=[coordinate]
   \setlength{\unitlength}{0.5cm}
   \nextwave{} {#1}
\end{tikzpicture}}

\newcommand\pptracebeta{\singlewave{\bit{0}{2} \bit{1}{1} \bit{0}{3} \bit{1}{1} \bit{0}{4} \bit{1}{1} \bit{0}{5} \bit{1}{1} \bit{0}{3}}}
\newcommand\pptraceedel{\singlewave{\bit{0}{4} \bit{1}{1} \bit{0}{3} \bit{1}{1} \bit{0}{4} \bit{1}{1} \bit{0}{5} \bit{1}{1} \bit{0}{1}}}
\newcommand\pptracebdel{\singlewave{\bit{0}{5} \bit{1}{1} \bit{0}{2} \bit{1}{2} \bit{0}{1} \bit{1}{2} \bit{0}{1} \bit{1}{1} \bit{0}{5} \bit{1}{1}}}
\newcommand\pptracefdel{\singlewave{\bit{0}{10} \bit{1}{1} \bit{0}{8} \bit{1}{1} \bit{0}{1}}}

\newcommand\maximaltracesk{
\begin{tikzpicture}[draw=black, yscale=.4,xscale=0.2]
 \tikzstyle{time}=[coordinate]
   \setlength{\unitlength}{0.5cm}
    \nextwave{$\beta$} \bit{0}{2} \bit{1}{2} \bit{0}{5}
    \nextwave{$KO^{\le 4} \beta$} \bit{0}{4} \bit{1}{3} \bit{0}{3}
    \nextwave{A (Maximal)} \bit{0}{4} \bit{1}{3} \bit{0}{3}
    \nextwave{A (Non-Maximal)} \bit{0}{5} \bit{1}{1} \bit{0}{4}
\end{tikzpicture}}

\newcommand\dlocal{trace}
\newcommand\dglobal{system}
\newcommand\dlocally{trace}
\newcommand\dglobally{system}
\newcommand\dLocal{Trace}
\newcommand\dGlobal{System}

\newcommand{\completeness}[1]{\colorbox{yellow!40}{$\displaystyle #1$}}
\newcommand{\correctness}[1]{\colorbox{orange!40}{$\displaystyle #1$}}
\newcommand{\diagnosability}[1]{\colorbox{cyan!40}{$\displaystyle #1$}}
\newcommand{\maximality}[1]{\colorbox{red!40}{$\displaystyle #1$}}

\newcommand{\delay}{\ensuremath{d}}

\begin{document}
\title[Formal Design of Asynchronous FDI using Temporal Epistemic Logic]
      {Formal Design of Asynchronous \\
        Fault Detection and Identification Components \\
        using Temporal Epistemic Logic}

\author[M.~Bozzano]{Marco Bozzano}
\address{\vspace{-18 pt}}
\author[A.~Cimatti]{Alessandro Cimatti}
\address{\vspace{-18 pt}}
\author[M.~Gario]{Marco Gario}
\address{\vspace{-18 pt}}
\author[S.~Tonetta]{Stefano Tonetta}
\address{Fondazione Bruno Kessler, Trento, Italy}
\email{\{bozzano, cimatti, gario, tonettas\}@fbk.eu}

%% mandatory lists of keywords and classifications:
\keywords{Fault Detection and Identification; Diagnoser Synthesis;
Model Checking; Temporal Epistemic Logic}

%%%%%%%%%%%%%%%%%%%%%%%%%%%%%%%%%%%%%%%%%%%%%%%%%%%%%%%%%%%%%%%%%%%%%%%%%%%

%% \input{files/cover_letter.tex}

\begin{abstract}
Autonomous critical systems, such as satellites and space rovers, must
be able to detect the occurrence of faults in order to ensure correct
operation. This task is carried out by Fault Detection and
Identification (FDI) components, that are embedded in those systems
and are in charge of detecting faults in an automated and timely
manner by reading data from sensors and triggering predefined alarms.

The design of effective FDI components is an extremely hard problem,
also due to the lack of a complete theoretical foundation, and of
precise specification and validation techniques.

%% Contribution

In this paper, we present the first formal approach to the design of
FDI components for discrete event systems, both in a synchronous and
asynchronous setting.  We propose a logical language for the
specification of FDI requirements that accounts for a wide class of
practical cases, and includes novel aspects such as maximality and
trace-diagnosability. The language is equipped with a clear semantics
based on temporal epistemic logic, and is proved to enjoy suitable
properties. We discuss how to validate the requirements and how to
verify that a given FDI component satisfies them. We propose an
algorithm for the synthesis of correct-by-construction FDI components,
and report on the applicability of the design approach on an industrial
case-study coming from aerospace.
\end{abstract}

%%% Original
\maketitle

\section{Introduction}
% SCOPE OF THE PAPER

The operation of complex critical systems (e.g., trains, satellites,
cars) increasingly relies on the ability to detect when and which
faults occur during operation. This function, called Fault Detection
and Identification (FDI), provides information that is vital to drive
the containment of faults and their recovery. This is especially true
for fail-operational systems, where the occurrence of faults should
not compromise the ability to carry on critical functions, as opposed
to fail-safe systems, where faults are typically handled by going to a
safe state.
FDI is often carried out by dedicated modules, called FDI components,
running in parallel with the system.  An FDI component, hereafter also
referred to as a diagnoser, processes sequences of observations, made
available by predefined sensors, and is required to trigger a set of
predefined alarms in a timely and accurate manner. The alarms are then
used by recovery modules to guarantee the survival of the system
without requiring external control.
Faults are often not directly observable. Their occurrence can only be
inferred by observing the effects that they have on the observable
parts of the system. Moreover, faults may have complex dynamics, and
may interact with each other in complex ways.

For these reasons, the design of FDI components is a very challenging
task, and also a practical problem, as witnessed by multiple
Invitations To Tender issued by the European Space
Agency~\cite{AUTOGEF-ITT,FAME-ITT,HASDEL-ITT}.  The current
methodologies lack a comprehensive theoretical foundation, and do not
provide clear and effective specification and validation techniques
and tools. Most approaches asses the quality of an FDI component based
on simulation and quantitative analysis~\cite{feldman2010empirical},
that do not start from a specification of the behavior the the FDI
needs to satisfy. This leads to a uniform treatment of all faults, while
in general some faults are more important then others, and in many
cases we are not interested in the specific fault characteristics but
only to know that the fault occurred in a given part of the system
(isolation). As a consequence, the design often results in very
conservative assumptions, so that the overall system features
sub-optimal behaviors, and it is not trusted during critical phases.

% OBJECTIVE OF THE PAPER/ CONTRIBUTIONS OF THE PAPER

The goal of this paper is to propose a formal foundation to support
the design of FDI components. We provide a way to specify FDI
components, and cover the following problems: (i) validation of an FDI
component specification, (ii) verification of a given FDI component
with respect to a given specification, and (iii) automated synthesis
of an FDI component from a given specification.

The specification of an FDI component is tackled by introducing a
\emph{pattern-based} language. Intuitively, an FDI component is
specified by stating the observable signals (the inputs of
the FDI component), the desired alarms (in terms of the
unobservable state), and by defining the relation between the two.
The language supports various forms of delay (exact, finite, bounded)
between the occurrence of faults and the raising of the corresponding
alarm.
The patterns are given a formal semantics expressed in terms of
epistemic temporal logic~\cite{HalpernVardi1989}, where the
\emph{knowledge} operator is used to express the certainty of a
condition, based on the available observations.
The formalization encodes properties such as \emph{alarm correctness}
and \emph{alarm completeness}. Correctness states that whenever an
alarm is raised by the FDI component, then its associated triggering condition
did occur; completeness states that if an alarm is not raised, then
either the associated condition did not occur, or it would have been
impossible to detect it, given the available observations.
Moreover, we precisely
characterize two aspects that are important for the specification of
FDI requirements. The first one is the \emph{diagnosability} of the
plant, i.e., whether the sensors convey enough information to detect
the required conditions. We explain how to deal with non-diagnosable
plants by introducing a more fine-grained concept of
\emph{\dlocal\ diagnosability}, where diagnosability is localized to
individual traces. Most of the state of the art focuses on the fact
that the system is diagnosable for any execution. However, in
practice, this is rarely the case, since usually the plant is
diagnosable in many situations but not in all of them. The classic
example is the one of a burnt light-bulb, of which we cannot say
anything until we try to turn it on. In this case, we would like to
build a diagnoser that can raise the alarm whenever there is no
ambiguity on whether the light bulb is burnt.  Therefore, we introduce
the concept of \dlocal\ diagnosability, intuitively accepting the fact
that the plant might not always be diagnosable.

The second important concept that we introduce is \emph{maximality}. A
diagnoser is maximal if it is able to raise an alarm as soon as and
whenever possible. This, in particular, means that in all traces that
are diagnosable, a maximal diagnoser needs to raise the alarm.

The approach provides a full account of synchronous and asynchronous
perfect-recall semantics for the epistemic operator. We show that the
specification language correctly captures the formal semantics and we
clearly define the relation between diagnosability, maximality and
correctness.

Within our setting, the validation of a diagnoser specification is
reduced to validity checking in temporal epistemic logic, while the
verification of a given diagnoser is mapped to model checking for a
temporal epistemic logic. As for synthesis, we propose an algorithm
that is proved to generate correct-by-construction diagnosers.

% INDUSTRIAL IMPACT

From the practical standpoint, the applicability of the design
approach has been demonstrated on two projects funded by the European
Space Agency~\cite{AUTOGEF-WEBSITE,FAME-WEBSITE}. The paper actually
provides the conceptual foundation underlying a design
tool-set~\cite{AUTOGEF-DASIA,IMBSA-regular,IMBSA-tool}, which has been
applied to the specification, verification and synthesis of an FDI
component for a satellite.

Finally, please note the deep difference between the design of FDI
components and most diagnosis~\cite{deKleer2004} approaches. In most
settings, diagnosis systems can benefit from powerful computing
platforms.  Partial diagnoses are typically acceptable, and can be
complemented by further (post-mortem) inspections. This is typical of
approaches that rely on logical reasoning engines (e.g., SAT
solvers~\cite{Grastien2007}).  Other
approaches~\cite{Huang2005,Sampath95,Schumann2004} rely on knowledge
compilation to reduce the on-line complexity.
An FDI component, on the contrary, runs on-board (as part of the
on-line control strategy), and is subject to restrictions of various
nature, such as timing and computation power. FDI design thus requires
a deeper theory, which accounts for the issues of delay in raising the
alarms, \dlocal\ diagnosability, and maximality. Moreover, it becomes
crucial to be able to verify and certify the effectiveness of the
system, since it might not be possible to change it after deployment.

% STRUCTURE OF THE PAPER
This paper is structured as follows.
Section~\ref{sec-background} provides some introductory
background and introduces our running example.
Section~\ref{sec-asl} formalizes the notion of FDI.
Section~\ref{sec-epistemic} presents the specification language.
In Section~\ref{sec-validation-verification}, we discuss how to validate
the requirements, and how to verify an FDI component with respect to
the requirements.
In Section~\ref{sec-fdi-synthesis}, we present an algorithm for the
synthesis of correct-by-construction FDI components.
The results of evaluating our approach in an industrial setting are
presented in Section~\ref{sec-industrial}.
Section~\ref{sec-related} compares our work with previous
related works.
In Section~\ref{sec-conclusions}, we draw some conclusions and outline
the directions for future work.

\section{Background}
\label{sec-background}

\subsection{Labeled Transition Systems}

In order to model the plant and the FDI, we use a symbolic
representation of \emph{Labeled Transition Systems} (LTS). Control
locations and data are represented by variables, while sets of states
and transitions are represented by formulas, and transitions are
labeled with explicit events.

Given a set of variables $X$ and a (finite) domain $\mathcal{U}$ of
values, an assignment to $X$ is a mapping from the set $X$ to the set
$\mathcal{U}$. We use $\Sigma(X)$ to denote the set of assignments to
$X$.
%\todo{R3: Give a formal definition of an assignment}.
Given an assignment $a\in\Sigma(X)$ and $X_1\subseteq X$, we use
$a_{|X_1}$ to denote the projection of $a$ over $X_1$. We use
$\mathcal{F}(X)$ to denote the set of propositional formulas over $X$.

\begin{defi}[LTS]
\label{def:LTS}
A \emph{Labeled Transition System} is a tuple
$S=\mktuple{V,E,I,\mathcal{T}}$, where:
\begin{itemize}
\item
$V$ is the set of state variables;
\item
$E$ is the set of events;
\item
$I\in\mathcal{F}(V)$ is a formula over $V$ defining the initial
  states;
\item
$\mathcal{T}: E\rightarrow \mathcal{F}(V\cup V')$ maps an event $e\in
  E$ to a formula over $V$ and $V'$ defining the transition relation
  for $e$ (with $V'$ being the next version of the state variables).
\end{itemize}
\end{defi}

A \emph{state} $s$ is an assignment to the state variables $V$ (i.e.,
$s\in\Sigma(V)$). We denote by $s'$ the corresponding assignment to
$V'$.
A transition labeled with $e$ is a pair of states $\mktuple{s,s'}$
such that $s,s'\models \mathcal{T}(e)$.
A \emph{trace} of $S$ is a sequence
$\sigma=s_0,e_0,s_1,e_1,s_{2},\ldots$ alternating states and event
such that $s_0$ satisfies $I$ and, for each $k\geq 0$,
$\mktuple{s_k,s_{k+1}}$ satisfies
$\mathcal{T}(e_k)$.
Note that we consider infinite traces only, and w.l.o.g.\ we assume the
system to be dead-lock free.
Given $\sigma=s_0,e_0,s_1,e_1,s_{2},\ldots$ and an integer $k\geq
0$, we denote by $\sigma^k$ the finite prefix
$s_0,e_0,\ldots,s_{k}$ of $\sigma$ containing the first $k+1$
states.  We denote by $\sigma[k]$ the $k+1$-th state $s_k$.
We say that $s$ is \emph{reachable} in $S$ iff there exists a trace
$\sigma$ of $S$ such that $s=\sigma[k]$ for some $k\geq 0$.

We say that $S$ is deterministic iff:
\begin{enumerate}
\item[(i)] there is one initial state (i.e., there exists a state $s$
  such that $s\models I$ and, for all $t$, if $t\models I$, then $s=t$);
\item[(ii)] for every reachable  state $s$, for every event $e$, there
  is one successor (i.e., there exists $s'$ such that
  $\mktuple{s,s'}\models \mathcal{T}(e)$ and, for all $t'$, if
  $\mktuple{s,t'}\models \mathcal{T}(e)$, then $s'=t'$).
\end{enumerate}

\begin{defi}[Synchronous Product]
\label{def:sync-product}
Let
\[S^1=\mktuple{V^1,E^1,I^1,\mathcal{T}^1}
\text{\ and \ } S^2=\mktuple{V^2,E^2,I^2,\mathcal{T}^2}\] be two
transition systems with $E^1=E^2=E$.
We define the \emph{synchronous product} $S^1\times S^2$ as the
transition system $\mktuple{V^1 \cup V^2, E, I^1 \wedge I^2,
  \mathcal{T}}$ where, for every $e\in E$,
$\mathcal{T}(e)=\mathcal{T}^1(e)\wedge \mathcal{T}^2(e)$.
Every state $s$ of $S^1\times S^2$ can be considered as the product
$s^1\times s^2$ such that $s^1=s_{|V^1}$ is a state of $S^1$ and
$s^2=s_{|V^2}$ is a state of $S^2$.
Similarly, every trace $\sigma$ of $S^1\times S^2$ can be considered
as the product $\sigma^1\times \sigma^2$ where $\sigma^1$ is a trace
of $S^1$ and $\sigma^2$ is a trace of $S^2$.
\end{defi}

\begin{defi}[Asynchronous Product]
\label{def:async-product}
Let
\[S^1=\mktuple{V^1,E^1,I^1,\mathcal{T}^1}
\text{\ and \ }
S^2=\mktuple{V^2,E^2,I^2,\mathcal{T}^2}\]
be two transition systems.
We define the
\emph{asynchronous product} $S^1\otimes S^2$ as the transition system
$\mktuple{V^1 \cup V^2, E^1\cup E^2, I^1 \wedge I^2, \mathcal{T}}$ where:
\begin{itemize}
\item
for every $e\in E^1\setminus E^2$,
$\mathcal{T}(e)=\mathcal{T}^1(e)\wedge \textit{frame}(V^2\setminus V^1)$.
\item
for every $e\in E^2\setminus E^1$,
$\mathcal{T}(e)=\mathcal{T}^2(e)\wedge \textit{frame}(V^1\setminus V^2)$.
\item
for every $e\in E^1\cap E^2$,
$\mathcal{T}(e)=\mathcal{T}^1(e)\wedge\mathcal{T}^2(e)$.
\end{itemize}
\noindent
where $\textit{frame}(X)$ stands for $\bigwedge_{x\in X}x'=x$ and is
used to represent the fact that while one transition system moves on a local
event, the other transition system does not change its local state variables.
Every state $s$ of $S^1\otimes S^2$ can be considered as the product
$s^1\otimes s^2$ such that $s^1=s_{|V^1}$ is a state of $S^1$ and
$s^2=s_{|V^2}$ is a state of $S^2$.
If either $S^1$ or $S^2$ is deterministic, also every trace $\sigma$
of $S^1\otimes S^2$ can be considered as the product $\sigma^1\otimes
\sigma^2$ where $\sigma^1$ is a trace of $S^1$ and $\sigma^2$ is a
trace of $S^2$ (more in general, the product of two traces produces a
set of traces due to different possible interleavings).
\end{defi}
In general, composing two systems can reduce the behaviors of each
system and introduce deadlocks. However, given two systems that do
not share any state variable (e.g., the diagnoser and the plant), if
one of the systems is deterministic (the diagnoser) then it cannot
alter the behavior of the second (the plant).

Notice that the synchronous product coincides with the asynchronous
case when the two sets of events coincide.

\subsection{Linear Temporal Logic}
We now present a Linear Temporal Logic extended with past
operators~\cite{LTL,PastLTL,LTLPast}, in the following simply referred to as
LTL. A formula in LTL over variables $V$ and events $E$ is defined as
\[ \beta ::= p \ | \ e \ |
          \ \beta \land \beta \ |
          \ \lnot \beta |
	      \ O \beta \ |
          \ Y \beta \ |
          \ \beta S \beta \ |
          \ G \beta \ |
          \ F \beta \ |
          \ X \beta \ |
          \ \beta U \beta
\]
\noindent where $p$ is a predicate over $\mathcal{F}(V)$ and $e\in
E$.  Intuitively, $p$ are the propositions over the state of the LTS,
while $e$ represents an event.

Given a trace $\sigma=s_0,e_0,s_1,e_1,s_{2},\ldots$, the semantics of
LTL is defined as follows:
\begin{itemize}
\item[-] $\sigma, i \models p$ iff $s_i \models p$
\item[-] $\sigma, i \models e$ iff $e_i=e$
\item[-] $\sigma, i \models \beta_1 \land \beta_2$ iff $\sigma, i \models \beta_1$ and
  $\sigma, i \models \beta_2$
\item[-] $\sigma, i \models \lnot \beta$ iff $\sigma, i \not \models
  \beta$
\item[-] Once: $\sigma, i \models O \beta$ iff $\exists j \le i .\ \sigma, j
  \models \beta$
\item[-] Yesterday: $\sigma, i \models Y \beta$ iff $i > 0$ and $\sigma, i-1 \models
  \beta$
\item[-] Since: $\sigma, i \models \beta_1 S \beta_2$ iff there exists $j\leq
  i$ such that $\sigma, j\models \beta_2$ and for all $k$, $j<k\leq i$,
  $\sigma, k \models \beta_1$
\item[-] Finally: $\sigma, i \models F \beta$ iff $\exists j \ge i .\ \sigma, j
  \models \beta$
\item[-] Globally: $\sigma, i \models G \beta$ iff $\forall j \ge i .\ \sigma, j
  \models \beta$
\item[-] Next: $\sigma, i \models X \beta$ iff $\sigma, i+1 \models \beta$
\item[-] Until: $\sigma, i \models \beta_1 U \beta_2$ iff there exists $j\geq
  i$ such that $\sigma, j\models \beta_2$ and for all $k$, $i\leq k< j$,
  $\sigma, k \models \beta_1$.\medskip
\end{itemize}

\noindent Given an LTS $S=\mktuple{V,E,I,\mathcal{T}}$, $S\models \beta$ iff for
every trace $\sigma$ of $S$, $\sigma, 0\models \beta$.

Notice that $Y \beta$ is always false in the initial
state, and that we use a reflexive semantics for the operators $U$, $F$, $G$,
$S$ and $O$. We use the abbreviations $Y^n \beta = Y Y^{n-1}
\beta$ (with $Y^0 \beta = \beta$), $O^{\le n} \beta = \beta \lor Y \beta
\lor \cdots \lor Y^n \beta$ and $F^{\le n} \beta = \beta \lor X \beta \lor
\cdots \lor X^n \beta$.

\subsection{Partial Observability}
\label{sec:partial-observability}
A \emph{partially observable} LTS is an LTS
$S=\mktuple{V,E,I,\mathcal{T}}$ extended with a set $E_o\subseteq E$
of observable events.

We consider here only observations on events. In practice, observation
on states are common and relevant. However, dealing with them in the
asynchronous setting makes the formalism less clear. Therefore, we
limit ourselves to observations on events and whenever observations on
state variables are needed, such as sensor readings, we incorporate
them in the events as done in~\cite{Sampath96}.

The observable part of the prefix $\sigma^k$ of a trace $\sigma$ is
defined recursively as follows: $\Obs(\sigma^0)=\epsilon$ (empty sequence); if
$e\in E_o$, then $\Obs(\sigma^k,e,s)=\Obs(\sigma^k),e$; if
$e\not\in E_o$, then $\Obs(\sigma^k,e,s)=\Obs(\sigma^k)$.

\begin{defi}[Observation Point]\label{def-obspoint}
We say that $i$ is an \emph{observation point} for $\sigma$, denoted
by $\ObsPoint(\sigma,i)$,
iff the last event of $\sigma^i$ is observable, i.e., iff
$\sigma^i=\sigma',e,s$ for some $\sigma', e, s$ and $e \in E_o$.
\end{defi}

The notion of two traces being observationally equivalent
requires that the two traces end both or neither in an observation point.
This captures the idea that a trace ending in an observation point can
be distinguished from the same trace extended with local unobservable
steps. In other terms, an observer can distinguish the instant in
which it is observing and an instant right after.
\begin{defi}[Observational Equivalence]
We say that $((\sigma_1, i), (\sigma_2, j)) \in \ObsEquiv$ if and only if:
\begin{itemize}
\item[-] $\ObsPoint(\sigma_1,i)$ iff $\ObsPoint(\sigma_2,j)$, and
\item[-] $\Obs(\sigma_1^i)=\Obs(\sigma_2^j)$.
\end{itemize}
\end{defi}

\subsection{Temporal Epistemic Logic}

Epistemic logic has been used to describe and reason about knowledge
of agents and processes. There are several ways of extending epistemic
logic with temporal operators. We use the logic
\klone~\cite{HalpernVardi1989}, extended with past operators.
A formula in \klone\ is defined as
\[ \beta ::= p \ |
          \ e \ |
          \ \beta \land \beta \ |
          \ \lnot \beta |
	  \ O \beta \ |
          \ Y \beta \ |
          \ \beta S \beta \ |
          \ F \beta \ |
          \ X \beta \ |
          \ \beta U \beta \ |
          \ G \beta \ |
          \ K \beta
\]

\klone\ can be seen as extension of LTL with past operators, with the
addition of the epistemic operator $K$.
The intuitive semantics of $K \beta$ is that the reasoner \emph{knows}
that $\beta$ holds in a state of a trace $\sigma$, by using only the
observable information. This means that $K \beta$ holds iff $\beta$
holds in all situations that are observationally equivalent.
Therefore, while in LTL the interpretation of a formula is local to a
single trace, in \klone\ the semantics of the $K$ operator quantifies
over the set of indistinguishable traces. Given a trace $\sigma_1$ of
a partially observable LTS, the semantics of $K$ is formally defined
as:
\[ \sigma_1,i \models K \beta \text{ iff }
   \forall \sigma_2, \forall j
   .\ \text{ if }((\sigma_1,i),(\sigma_2,j)) \in\ObsEquiv \text{ then
   } \sigma_2, j \models \beta.
\]

$K\beta$ holds at time $i$ in a trace $\sigma_1$ iff $\beta$ holds in
all traces that are observationally equivalent to $\sigma_1$ up to
time $i$. Note that, due to the asynchronous nature of the
observations, two traces of different length might lead to the same
observable trace.
This definition implicitly forces \emph{perfect-recall} in the
semantics of the epistemic operator, since we define the epistemic
equivalence between traces and not between states.

In many situations, we are interested in considering formulas only
at observation points. We do so by introducing the following
abbreviation.
\begin{defi}[Observed]\label{def-observed}
If $E_o$ is the set of observable events, given a formula $\phi$, we
use $\strong{\phi}$ (read ``Observed $\phi$'') as abbreviation for
$\phi\wedge Y \bigvee_{e\in E_o} e$.
\end{defi}

\subsection{Running Example}
\label{sec:running-example}
\begin{figure}[ht]
  \begin{center}
    \vspace{-1em}
    \resizebox{1\textwidth}{!}{
      \begin{tikzpicture}
  [
    component/.style={draw=black, thick, rectangle, minimum size=1.4cm},
    switch/.style={draw=black, thick, circle},
    conn/.style={thick,solid,-latex},
    every node/.style={node distance=1.5cm},
    power/.style={draw=blue},
    data/.style={draw=red},
    ctrl/.style={draw=green},
  ]

  \node[component] (G1) {Generator 1};
  \node[component] (G2) [below=of G1] {Generator 2};

  \node[component] (B1) [right=2.5cm of G1] {Battery 1};
  \node[component] (B2) [below=of B1] {Battery 2};

  \node[component] (S1) [right=of B1] {Sensor 1};
  \node[component] (S2) [below=of S1] {Sensor 2};

  \node[switch] (sw) at ($(B1.south east) +(0.75, -0.75)$) {Switch};

  \coordinate (E1) at ($(S1.east) + (2.5, 0)$);
  \coordinate (E2) at ($(S2.east) + (2.5, 0)$);

  \path (G1) edge [conn,power] node [anchor=south] {Generator IN} (B1);
  \path (G2) edge [conn,power] node [anchor=north] {Generator IN} (B2);

%  \path (B1) edge [conn] (S1);
%  \path (B2) edge [conn] (S2);

  \path (B1) edge [conn,dashed,power]  (S2);
  \path (B2) edge [conn,dashed,power]  (S1);

  \path (S1) edge [conn,data] node [anchor=south] {Sensor OUT} (E1);
  \path (S2) edge [conn,data] node [anchor=north] {Sensor OUT} (E2);

  \node (mdsel) at ($(sw) + (0, 3.5)$) {Mode Selector};
  \draw[conn,ctrl] (mdsel.south) -- (sw) ;

  \draw[thick,draw=black] ($(E1) + (0,1.0)$) rectangle node {Device} ($(E2.south east) + (2, -1.0)$);
  \draw[dashed, very thin] ($(B1.north west) + (-0.5,0.5)$) rectangle ($(S2.south east) + (0.5, -0.5)$);

  % Legenda
  \coordinate (lpower) at ($(E1.east) + (2.5, -0.5)$);
  \coordinate (lcontrol) at ($(lpower) + (0, -0.5)$);
  \coordinate (ldata) at ($(lcontrol) + (0, -0.5)$);

  \node[right] (lpower-end) at ($(lpower) + (0.5, 0)$) {Power};
  \node[right] (lcontrol-end) at ($(lcontrol) + (0.5, 0)$) {Control};
  \node[right] (ldata-end) at ($(ldata) + (0.5, 0.0)$) {Data};
        
  \draw[conn,power] (lpower) -- (lpower-end.west);
  \draw[conn,ctrl] (lcontrol) -- (lcontrol-end.west);
  \draw[conn,data] (ldata) -- (ldata-end.west);

\end{tikzpicture}
    }
    \caption{Running Example (Battery Sensor System)}
    \label{fig:BatterySensor}
  \end{center}
\end{figure}

\newcommand{\generatorlts}{
\begin{tikzpicture}[->,>=stealth',shorten >=1pt,
                    auto,node distance=3cm,
                    semithick]
  \tikzstyle{every state}=[draw=black]

  \node[initial,state] (Gon) {On};
  \node[state] (Goff) [right of=Gon] {Off};
  \path (Gon)  edge [loop above] node {}                 (Gon);
  \path (Goff) edge [loop above] node {}                (Goff);
  \path (Gon)  edge              node {Fault \& Off} (Goff);
\end{tikzpicture}}

\newcommand{\sensorlts}{
\begin{tikzpicture}[->,>=stealth',shorten >=1pt,
    auto,node distance=3cm, semithick]
  \tikzstyle{every state}=[draw=black]

  \node[initial,state] (GoodN) {Good (N)};
  \node[state] (BadN) [right of=GoodN] {Bad (N)};
  \node[state] (BadF) [right of=BadN] {Bad (F)};

  \path (GoodN)
  edge [loop above] node {Value=Good} (GoodN)
  edge [bend right] node [anchor=north] {Batt.c = 0} (BadN);
  \path (GoodN.south) edge [bend right] node [anchor=north] {Fault} (BadF);
  \path (BadN)
  edge [loop above] node {Value=Bad} (BadN)
  edge              node {Fault} (BadF)
  edge [bend right] node [anchor=south] {Batt.c $>$ 0} (GoodN);
  \path (BadF)
  edge [loop below] node {Value=Bad} (BadF);
\end{tikzpicture}}

\newcommand{\devicelts}{
\begin{tikzpicture}[->,>=stealth',shorten >=1pt,
                    auto,node distance=2.7cm, semithick]
  \tikzstyle{every state}=[draw=black]

  \node[initial,state] (B4) {On};
  \node[state] [right of=B4] (B3) {On};
  \node[state] [right of=B3] (Off) {Off};

  \path (B4)
  edge [loop above] node [] {\emph{stay}
                                             %% S1 = Good $\land$ \\
                                             %% S2 = Good $\land$ \\
                                             %% $\Delta$=Zero
  } (B4)
  edge node [] {\emph{degrade}
  %  (S1=Bad $\lor$ S2=Bad) $\land$ \\ $\Delta$=Non-Zero
  } (B3);
  \path (B3) [loop below] edge node {$x' = x-1$} (B3);
  \path (B3) edge node {x=0 $\land$ Off} (Off);
  \path (Off) [loop right] edge node {} (Off);
\end{tikzpicture}}

%% where \emph{stay} is defined as $S1.Value = S2.Value \land
%% Delta=Zero$, while \emph{start} represent a discrepancy in the reading
%% from the sensor that will eventually lead to the device stopping:
%% $(S1.Value \not = S2.Value) \land Delta=Non-Zero$. The
%% values of the sensors are not observable, but their difference is
%% observable via the $Delta$ variable.

\newcommand{\labelN}{Nominal\\}
\newcommand{\labelF}{Faulty\\}
\newcommand{\labelP}{Primary\\}
\newcommand{\labelO}{Offline\\}
\newcommand{\labelD}{Double\\}
\newcommand{\labelC}{Charging}
\newcommand{\labelnC}{Not Charging}

\newcommand{\batterylts}{
\begin{tikzpicture}[->,>=stealth',shorten >=1pt,
    auto,node distance=3.7cm, semithick,
    text width=2.3cm,
   ]
  \tikzstyle{every state}=[draw=black]
  \tikzstyle{observableT}=[color=black]

  \node[state] (NPC) {\labelN \labelP \labelC};
  \node[text width=0.8cm] (startN) [left = 0.7cm of NPC] {start};
  \node[state] (NOC) [below of=NPC] {\labelN \labelO \labelC};
  \node[state] (NDC) [below of=NOC]{\labelN \labelD \labelC};

  \node[state] (FPC) [right of=NPC] {\labelF \labelP \labelC};
  \node[state] (FOC) [below of=FPC] {\labelF \labelO \labelC};
  \node[state] (FDC) [below of=FOC] {\labelF \labelD \labelC};

  \node[state] (FPD) [right of=FPC] {\labelF \labelP \labelnC};
  \node[state] (FOD) [below of=FPD] {\labelF \labelO \labelnC};
  \node[state] (FDD) [below of=FOD] {\labelF \labelD \labelnC};

  \node[state] (NPD) [right of=FPD] {\labelN \labelP \labelnC};
  \node[state] (NOD) [below of=NPD] {\labelN \labelO \labelnC};
  \node[state] (NDD) [below of=NOD] {\labelN \labelD \labelnC};

  \path (startN) edge [observableT] node {} (NPC);

  \path (NPC)
%%  edge [loop above] {} (NPC)
  edge node {} (FPC)
  edge [observableT] node {} (NOC)
  edge [observableT, bend right=40] node {} (NDC)
  edge [bend left] node {} (NPD);

  \path (FPC)
%%  edge [loop above] {} (FPC)
  edge node {} (FPD)
  edge [observableT] node {} (FOC)
  edge [observableT, bend right=40] node {} (FDC);

  \path (FPD)
%%  edge [loop above] {} (FPD)
  edge node {} (FPC)
  edge [observableT] node {} (FOD)
  edge [observableT, bend left=40] node {} (FDD);

  \path (NPD)
%%  edge [loop above] {} (NPD)
  edge node {} (FPD)
  edge [bend right] node {} (NPC)
  edge [observableT] node {} (NOD)
  edge [observableT, bend left=40] node {} (NDD);

  \path (NOC)
%%  edge [loop above] {} (NOC)
  edge [observableT] node {} (NPC)
  edge [observableT] node {}  (NDC)
  edge node {} (FOC)
  edge [bend right] node {} (NOD);

  \path (FOC)
%%  edge [loop above] {} (FOC)
  edge [observableT] node {} (FPC)
  edge [observableT] node {} (FDC)
  edge node {} (FOD);

  \path (FOD)
%%  edge [loop above] {} (FOD)
  edge node {} (FOC)
  edge [observableT] node {} (FPD)
  edge [observableT] node {} (FDD);

  \path (NOD)
%%  edge [loop above] {} (NOD)
  edge node {} (FOD)
  edge [bend right] node {} (NOC)
  edge [observableT] node {} (NPD)
  edge [observableT] node {} (NDD);

  \path (NDC)
%%  edge [loop above] {} (NDC)
  edge [observableT, bend left=40] node {} (NPC)
  edge [observableT] node {}  (NOC)
  edge node {} (FDC)
  edge [bend right] node {} (NDD);

  \path (FDC)
%%  edge [loop above] {} (FDC)
  edge [observableT, bend left=40] node {} (FPC)
  edge [observableT] node {} (FOC)
  edge node {} (FDD);

  \path (FDD)
%%  edge [loop above] {} (FDD)
  edge node {} (FDC)
  edge [observableT, bend right=40] node {} (FPD)
  edge [observableT] node {} (FOD);

  \path (NDD)
%%  edge [loop above] {} (NDD)
  edge node {} (FDD)
  edge [bend left] node {} (NDC)
  edge [observableT, bend right=40] node {} (NPD)
  edge [observableT] node {} (NOD);
\end{tikzpicture}}

%% Each state has three labels. The first indicates \emph{N}ominal or
%% \emph{F}aulty situation. The second shows the mode: \emph{P}rimary, \emph{O}ffline
%% or \emph{D}ouble. The third indicates whether the battery is
%% \emph{C}harging or not ($\overline{C}$).

%% Each state has an additional self-loop (not in the picture) denoting
%% the update of the charge of the battery, following the update rule:
%% %
%% \[ charge' = (charge + recharge - (load + leak)) \text{ mod } C \]
%% %
%% \noindent where $C$ is the capacity, and the other variables depend on
%% the the state:
%% \begin{enumerate}
%% \item Charging: $recharge = 1$, Not Charging: $recharge = 0$
%% \item Primary: $load = 1$, Offline: $load = 0$, Double: $load = 2$
%% \item Nominal: $leak = 0$, Faulty: $leak = 2$
%% \end{enumerate}
%% %
%% Thus the charge of the battery can change from $+1$ (Nominal, Offline,
%% Charging) to $-4$ (Faulty, Double, Not Charging), while staying within
%% the bound $[0, Capacity]$.

%% Every time the update of the charge causes the charge to pass a
%% threshold, the transition raises the observable event: $Low$, $Mid$,
%% $High$. These events indicate when the charge of the battery is above
%% 20\%, 50\% and 80\%. All other transitions are not observable. These
%% transitions have been omitted from the figure to make it more
%% readable.

\newcommand{\switchlts}{
\begin{tikzpicture}[->,>=stealth',shorten >=1pt,
    auto,node distance=3cm, semithick,
   ]
  \tikzstyle{every state}=[draw=black]
  \node[state] (Primary) {Primary};
  \node(start) at ($(Primary.north)+(-1.4, +0.6)$) {};
  \node[text width=1.55cm, align=center, state] (S1) [left of=Primary] {Secondary 1};
  \node[text width=1.55cm, align=center, state] (S2) [right of=Primary] {Secondary 2};

  \path (start) edge node {start} (Primary);
  \path (Primary)
  edge node [anchor=south] {toS1} (S1)
  edge node {toS2} (S2);
  \path (S1) edge [loop left] node {} (S1);
  \path (S2) edge [loop right] node {} (S2);

\end{tikzpicture}}

%% \begin{itemize}
%% \item toS1: Mode=Secondary1 $\land$ Battery1.Double $\land$ Battery2.Offline
%% \item toS1: Mode=Secondary2 $\land$ Battery1.Offline $\land$ Battery1.Double
%% \end{itemize}

The \emph{Battery Sensor System} (BSS)
(Figure~\ref{fig:BatterySensor}) will be our running example. The BSS
provides a redundant reading of the sensors to a device. Internal
batteries provide backup in case of failure of the external power
supply. The safety of the system depends on both of the sensors
providing a correct reading. The system can work in three different
operational modes: \emph{Primary}, \emph{Secondary 1} and
\emph{Secondary 2}. In Primary mode, each sensor is powered by the
corresponding battery. In the Secondary modes, instead, both sensors
are powered by the same battery; e.g., during Secondary 1, both Sensor
1 and Sensor 2 are powered by Battery 1. The Secondary modes are used
to keep the system operational in case of faults. However, in the
secondary modes, the battery in use will discharge faster.

We consider two possible recovery actions: i) Switch Mode, or ii) Replace
the Battery-Sensor Block (the dotted block in
Figure~\ref{fig:BatterySensor}). In order to decide which recovery to
apply, we are going to define a set of requirements connecting the
faults to alarms. The faults and observable information of the system
are shown in Figure~\ref{fig:obs-faults-summary}.

This example is particularly interesting because we can define two
sources of delay: the batteries, and the device resilience to wrong
inputs. The batteries provide a buffer for supplying power to the
sensors. The size of this buffer is determined by the capacity of the
battery, the initial charge, and the discharge rate. For the device,
we assume that two valid sensor readings are required for optimal
behavior, however, we can work in degraded mode with only one valid reading
for a limited amount of time. The device will stop working
if both sensors are providing invalid readings, or if one sensor has
been providing an invalid reading for too long.

Both a synchronous and asynchronous version of this model are
possible. In the asynchronous model, we have an event for each
possible combination of observations (e.g., ``Mode Primary \& Battery
1 Low''). In the synchronous model, we also have an additional
observable event ($tick$) that represents the passing of time in the
absence of any observable event. This event forces the synchronization
of the plant with the diagnoser. The key difference between the
synchronous and asynchronous setting is the amount of information that
we can infer in this particular case. For example, if we know the
initial charge level of a battery, and we know its discharge rate
(given by the operational mode), then at each point in time we can
infer the current charge of the battery. By comparing our expectation
with the available information, we can detect when something is not
behaving as expected.  Unfortunately, there are practical settings in
which the assumption of synchronicity is not realistic. Therefore, our
approach accounts for both the synchronous and asynchronous models.

\begin{figure}[ht]
\resizebox{\textwidth}{!}{
\begin{tabular}{|l|l|}
\hline
\textbf{Observables} & \textbf{Possible Values} \\ \hline
Mode & Primary, Secondary 1, Secondary 2 \\
Battery Level \{1, 2\} & High, Mid, Low \\
Sensors Delta & Zero, Non-Zero ($|S1.Out - S2.Out| = 0 $) \\
Device Status & On, Off \\
\hline
\end{tabular}
\begin{tabular}{|l|l|}
\hline
\textbf{Component} & \textbf{Faults} \\ \hline
Generator & Off ($G1_{Off}, G2_{Off}$) \\
Battery   & Leak ($B1_{Leak}, B2_{Leak}$) \\
Sensor    & Wrong Output ($S1_{WO}, S2_{WO}$) \\
\hline
\end{tabular}}
\caption{Observables and Faults Summary}
\label{fig:obs-faults-summary}
\end{figure}

To provide a better understanding of how the running example behaves,
we provide the LTS of each of the
components. Figure~\ref{fig:BatterySensor-Gen-and-Switch-LTS} shows the
LTS of the generator and switch. We assume that the only way the
generator can turn off is if a fault event occurs, thus the model of
the generator is rather simple. Also the switch features a rather
simple model, where the labels \emph{toS1} and \emph{toS2} are defined
as:
\begin{itemize}
\item toS1: Mode=Secondary1 $\land$ Battery1.Double $\land$ Battery2.Offline
\item toS1: Mode=Secondary2 $\land$ Battery1.Offline $\land$ Battery1.Double
\end{itemize}
\noindent thus they drive the change in operational mode of the batteries.

\begin{figure}[ht]
  \begin{center}
    \begin{minipage}{0.3\textwidth}
      \generatorlts
    \end{minipage}\begin{minipage}{0.7\textwidth}
      \hspace{1.5cm}\resizebox{0.9\textwidth}{!}{\switchlts}
    \end{minipage}
    \caption{Generator (Left) and Switch (Right) LTS}
    \label{fig:BatterySensor-Gen-and-Switch-LTS}
  \end{center}
\end{figure}

Figure~\ref{fig:BatterySensor-Sensor-and-Device-LTS} shows two
slightly more complex components: the sensor and the device. The
sensor periodically outputs a good or a bad reading depending on the
state it is in. Notice that the transition from a good to a bad state
can occur either because of a fault (Wrong Output in Figure~\ref{fig:obs-faults-summary})
or because the battery connected to the sensor has no charge
($Batt.c=0$), notice, in particular, that both events are not
observable. The device instead has two main transitions. The
\emph{stay} is defined as $S1.Value = S2.Value \land Delta=Zero$,
while \emph{degrade} represents a discrepancy in the reading from the
sensor that will eventually lead to the device stopping: $(S1.Value
\not = S2.Value) \land Delta=$\emph{Non-Zero}. The values of the sensors are
not observable, but their difference is observable via the $Delta$
variable. Intuitively, the device has an intermediate state that works
as a buffer, before reaching the final \emph{Off} state.

\begin{figure}[ht]
  \begin{center}
    \begin{minipage}{0.5\textwidth}
      \resizebox{1.0\textwidth}{!}{\sensorlts}
    \end{minipage}\begin{minipage}{0.5\textwidth}
      \resizebox{0.9\textwidth}{!}{\devicelts}
    \end{minipage}
    \caption{Sensor (Left) and Device (Right) LTS}
    \label{fig:BatterySensor-Sensor-and-Device-LTS}
  \end{center}
\end{figure}

The most complex component, the battery, is presented in
Figure~\ref{fig:BatterySensor-Battery-LTS}. Vertical transitions
indicate a change in operational mode of the battery. The left half of
the LTS indicates that the generator is working and feeding the
battery (thus charging it) while the right half shows that the
battery is not charging. Additionally, the two central columns
describe the faulty behavior of the battery. This information is
represented also in each state.
Each state has an additional self-loop (not in the picture) denoting
the update of the charge of the battery, following the update rule:
\[ charge' = (charge + recharge - (load + leak)) \text{ mod } C \]
\noindent where $C$ is the capacity, and the other variables depend on
the state:
\begin{enumerate}
\item Charging: $recharge = 1$, Not Charging: $recharge = 0$
\item Primary: $load = 1$, Offline: $load = 0$, Double: $load = 2$
\item Nominal: $leak = 0$, Faulty: $leak = 2$
\end{enumerate}
Thus the charge of the battery can change from $+1$ (Nominal, Offline,
Charging) to $-4$ (Faulty, Double, Not Charging), while staying within
the bound $[0, Capacity]$.

Every time the update of the charge causes the charge to pass a
threshold, the transition raises the observable event: $Low$, $Mid$,
$High$. These events indicate when the charge of the battery is above
20\%, 50\% and 80\%. All other transitions are not observable. These
transitions have been omitted from the figure to make it more
readable.

\begin{figure}[ht]
  \begin{center}
    \resizebox{\textwidth}{!}{ \batterylts }
    \caption{Battery LTS}
    \label{fig:BatterySensor-Battery-LTS}
  \end{center}
\end{figure}

\section{Formal Characterization}
\label{sec-asl}

\subsection{Diagnoser}

In our general setting, a plant is connected to components for Fault
Detection and Isolation, and for Fault Recovery, as depicted in
Figure~\ref{fig:fdir-plant}. The role of FDI is to collect and analyze
the observable information from the plant, and to turn on suitable
alarms associated with (typically unobservable) relevant conditions.
The Fault Recovery component is intended to apply suitable
reconfiguration actions based on the alarms in input. Recovery is
beyond the scope of this work; we consider a \textit{system} composed
of the plant and the FDI component.

An FDI component (also called diagnoser in the following) is a machine
$D$ that synchronizes with observable traces of the plant $P$. $D$ has
a set $\alarms$ of alarms that are activated in response to the
monitoring of $P$. Different mechanisms to connect a diagnoser to a
plant are possible. In the synchronous case, the plant is assumed to
convey to the diagnoser information at a fixed rate (including state
sampling and values for event ports). This model is adopted, for
example, in~\cite{DBLP:conf/tacas/BozzanoCGT14,Cimatti2003}. In this
paper we focus on the more general model of asynchronous case, where
the diagnoser reacts to the observable events in the plant
\footnote{The relation between the synchronous and the asynchronous
  combination is discussed in Section~\ref{sec-sync-into-async}.}.

\begin{figure}[ht]
\resizebox{0.4\textwidth}{!}{\begin{tikzpicture}
  [
    plant/.style={rectangle,draw=black,thick,minimum size=2cm},
    sensors/.style={rectangle,draw=black,thick,rotate=90, minimum width=2.7cm},
    actuators/.style={rectangle,draw=black,thick,rotate=90, minimum width=2.7cm},
    conn/.style={->,thick,solid},
    fdi/.style={rectangle,draw=black,thick,minimum size=1cm,rounded corners=0.1cm},
    fir/.style={rectangle,draw=black,thick,minimum size=1cm,rounded
      corners=0.1cm},
    bigconn/.style={line width=0.8ex,-latex,rounded corners=0.7cm,
      postaction={draw,color=white,line width=0.6ex,
        shorten >=0.4ex,shorten <=.01ex}},
  ]

  \node[plant]  (plant) at (0, 0) {Plant};
  \node[sensors] (sensors) at (-2cm, 0)  {SENSORS};
  \node[actuators] (actuators) at (2cm, 0) {ACTUATORS};

  \coordinate (sensor-left) at ($(sensors.north east) + (-0.3, 0.4)$);
  \coordinate (actuator-right) at ($(actuators.south west) + (+0.3, -0.4)$);

  \draw[thick,draw=black,rounded corners=0.5cm] (sensor-left) rectangle node {} (actuator-right);

  %% FDIR
  \node[fdi] at (-0.9, 3.5) (fdi) {FDI};
  \node[fir] at (0.9, 3.5) (fir) {FIR};

  \coordinate (fdi-left) at ($(fdi.west) + (-0.3, 1)$);
  \coordinate (fir-right) at ($(fir.east) + (+0.3, -1)$);

  \draw[thick,draw=black,rounded corners=0.1cm] (fdi-left) rectangle node {} (fir-right);

%% Connections
\draw[conn] ($(plant.west) + (0,0.6)$) -- ($(sensors.south) + (0, 0.6)$);
\draw[conn] ($(plant.west) + (0,0.3)$) -- ($(sensors.south) + (0, 0.3)$);
\draw[conn] ($(plant.west) + (0,0)$) -- ($(sensors.south) + (0, 0)$);
\draw[conn] ($(plant.west) + (0,-0.3)$) -- ($(sensors.south) + (0,-0.3)$);
\draw[conn] ($(plant.west) + (0,-0.6)$) -- ($(sensors.south) + (0,-0.6)$);

\draw[conn] ($(actuators.north) + (0,0.6)$) -- ($(plant.east) + (0, 0.6)$);
\draw[conn] ($(actuators.north) + (0,0.3)$) -- ($(plant.east) + (0, 0.3)$);
\draw[conn] ($(actuators.north) + (0,0)$) -- ($(plant.east) + (0, 0)$);
\draw[conn] ($(actuators.north) + (0,-0.3)$) -- ($(plant.east) + (0,-0.3)$);
\draw[conn] ($(actuators.north) + (0,-0.6)$) -- ($(plant.east) + (0,-0.6)$);

\draw[conn] ($(fdi.east) + (0,0.3)$) -- ($(fir.west) + (0, 0.3)$);
\draw[conn] (fdi) -- (fir);
\draw[conn] ($(fdi.east) - (0,0.3)$) -- ($(fir.west) - (0, 0.3)$);

\draw[bigconn]
   (sensors) -- ($(sensors.north) + (-1.2,0)$) -- ($(sensors.north) + (-1.2,+3.5)$) -- (fdi.west);

\draw[bigconn] (fir) --  ($(actuators.south) + (1.2,3.5)$) --
                 ($(actuators.south) + (1.2,0)$) -- (actuators);

\end{tikzpicture}}
\caption{Integration of the FDIR and Plant}
\label{fig:fdir-plant}
\end{figure}

\sloppypar
\begin{defi}[Diagnoser]
Given a set $\alarms$ of alarms and a partially observable
plant $P=\mktuple{V^P,E^P,I^P,\mathcal{T}^P, E^P_o}$, a diagnoser is a
deterministic LTS $D(\alarms, P)=\mktuple{V^D,E^D,I^D,\mathcal{T}^D}$
such that $E^P_o = E^D$, $V^P \cap V^D =\emptyset$ and
$\alarms\subseteq V^D$.
\end{defi}

\noindent When clear from the context, we use $D$ to indicate
$D(\alarms,P)$.
We assume that the events of the diagnoser coincide with
the observable events of the plant. This means that the diagnoser
does not have internal transitions: every transition of the diagnoser
is associated with an observable transition of the plant.
We say that the alarm $A$ is triggered when $A$ is true after the
diagnoser synchronized with the plant (i.e., when $\strong{A}$ is
true).

Since the synchronous case is a particular case of the asynchronous
composition, in the rest of the paper we assume that the plant and
diagnoser are composed asynchronously: i.e., $D\otimes P$. Only
observable events are used to perform synchronization.

The choice of using a deterministic diagnoser is driven by the
following result, that makes it easier to understand how the diagnoser
will react to the plant:
\begin{defi}[Diagnoser Matching trace]\label{def-matchingtrace}
Given a diagnoser $D$ of $P$ and a trace $\sigma_P$ of $P$, the
diagnoser trace matching $\sigma_P$, denoted by $D(\sigma_P)$, is the
trace $\sigma$ of $D$ such that $\sigma\otimes \sigma_P$ is a trace of
$D\otimes P$.
\end{defi}
Note that the notion of diagnoser matching trace is well defined
because, since $D$ is deterministic, there exists one and only one
trace in $D$ matching $\sigma_P$.

\subsection{Detection, Identification, and Diagnosis Conditions}

The first element for the specification of the FDI requirements is
given by the conditions that must be monitored. Here, we distinguish
between detection and identification, which are the two extreme cases
of the diagnosis problem; the first deals with knowing whether a fault
occurred in the system, while the second tries to
identify the characteristics of the fault. Between these two cases
there can be intermediate ones: we might want to restrict the
detection to a particular sub-system, or identification among two
similar faults might not be of interest.

The \emph{detection} task is the problem of understanding when (at
least) one of the components has failed. The \emph{identification}
task tries to understand exactly which fault occurred.

In the BSS every component can fail. Therefore the detection problem
boils down to knowing that at least one of the generators, batteries
or sensors is experiencing a fault. For identification, instead, we
are interested in knowing whether a specific fault, (e.g., $G1_{Off}$)
occurred. There are also intermediate situations (sometimes called
\emph{isolation}), in which we are not interested in distinguishing
whether $G1_{Off}$ or $B1_{Leak}$ occurred, as long as we know that
there is a problem in the power-supply chain.

FDI components are generally used to recognize faults. However, there
is no reason to restrict our interest to faults. Recovery procedures
might differ depending on the current state of the plant, therefore,
it might be important to consider other unobservable information of
the system. For example, we might want to estimate the charge level
of a battery, or its discharge rate.

We call the condition of the plant to be monitored \emph{diagnosis
  condition}, denoted by $\beta$.  We assume that for any
point in time along a trace execution of the plant (and therefore also
of the system), $\beta$ is either true or false based on what happened
before that time point. Therefore, $\beta$ can be an atomic condition
(including faults), a sequence of atomic conditions, or Boolean
combination thereof. If $\beta$ is a fault, the fault must be
identified; if $\beta$ is a disjunction of faults, instead, it
suffices to perform the detection, without identifying the exact
fault.

\begin{figure}[ht]
\begin{tabular}{|l|l|}
\hline
\textbf{Diagnosis condition} & \textbf{Definition} \\ \hline
$\beta_{Generator 1}$, $\beta_{Generator 2}$ & $G1_{Off}$, $G2_{Off}$\\
$\beta_{Battery 1}$, $\beta_{Battery 2}$ & $B1_{Leak}$, $B2_{Leak}$ \\
$\beta_{PSU 1}$, $\beta_{PSU 2}$ & $G1_{Off} \lor B1_{Leak}$, $G2_{Off} \lor B2_{Leak}$ \\
$\beta_{Batteries}$ & $B1_{Leak} \lor B2_{Leak}$ \\
$\beta_{Sensor 1}$, $\beta_{Sensor 2}$ & $S1_{WO}$, $S2_{WO}$ \\
$\beta_{Sensors}$ & $S1_{WO} \lor S2_{WO}$ \\
$\beta_{BS}$ & $(S1_{WO} \lor S2_{WO}) \lor (B1_{Leak} \land
B2_{Leak})$\\
$\beta_{Seq}$ & $(B1_{Charge} < B2_{Charge}) \land O (B1_{Charge} \ge B2_{Charge}) $\\
$\beta_{Charging}$ &
$Y ((B1_{Charge} \le 0) \land Y ( B1_{Charge} > 0)$ \\
$\beta_{Depleted}$ & $(B1_{Charge} = 0) \lor (B2_{Charge} = 0)$ \\
\hline
\end{tabular}
\caption{Diagnosis conditions for the BSS}
\label{fig:bss-diagnosis-conditions}
\end{figure}

Figure~\ref{fig:bss-diagnosis-conditions} shows several examples of
diagnosis conditions for the BSS. Notice how we might be in complex
situations such as knowing if the Battery-Sensor block is working
($\beta_{BS}$) or knowing some information on the evolution of the
system ($\beta_{Seq}$, $\beta_{Charging}$). We use LTL operators to define
those diagnosis conditions, but in general, we require that a
diagnosis condition can be evaluated on a point in a trace by only
looking at the trace prefix.

\subsection{Alarm Conditions}
The second element of the specification of FDI requirements is the
relation between a diagnosis condition and the raising of an alarm.
This also leads to the definition of when the FDI is correct and
complete with regard to a set of alarms.

An \emph{alarm condition} is composed of two parts: the diagnosis
condition and the delay. The delay relates the time between the
occurrence of the diagnosis condition and the
corresponding alarm. Although it might be acceptable that the occurrence of a
fault can go undetected for a certain amount of time, it is important
to specify clearly how long this interval can be.
An alarm condition is a property of the system composed by the plant
and the diagnoser, since it relates a condition of the plant with an
alarm of the diagnoser. Thus, when we say that a diagnoser $D$ of $P$
satisfies an alarm condition, we mean that the traces of the system
$D\otimes P$ satisfy it.

Interaction with industrial experts led us to identify three patterns
of \emph{alarm conditions}, which we denote by
\edel{A}{\beta}{\delay}, \bdel{A}{\beta}{\delay}, and \fdel{A}{\beta}:\\

1. \edel{A}{\beta}{\delay} specifies that whenever $\beta$ is true, $A$ must
be triggered exactly $\delay$ steps later and $A$ can be triggered only if
$\delay$ steps earlier $\beta$ was true; formally, for any trace $\sigma$ of
the system, if $\beta$ is true along $\sigma$ at the time point $i$, then
$\trigger{A}$ is true in $\sigma[i+\delay]$ (Completeness); if $\trigger{A}$ is true in
$\sigma[i]$, then $\beta$ must be true in $\sigma[i-\delay]$ (Correctness).

2.
\bdel{A}{\beta}{\delay} specifies that whenever $\beta$ is true, $A$ must
be triggered within the next $\delay$ steps and $A$ can be triggered only
if $\beta$ was true within the previous $\delay$ steps; formally, for any
trace $\sigma$ of the system, if $\beta$ is true along $\sigma$ at the time
point $i$ then $\trigger{A}$ is true in $\sigma[j]$, for some $i \le j \le i+\delay$
(Completeness); if $\trigger{A}$ is true in $\sigma[i]$, then $\beta$ must be true
in $\sigma[j']$ for some $i-\delay \le j' \le i$
(Correctness).

3.
\fdel{A}{\beta} specifies that whenever $\beta$ is true, $A$ must be
triggered in a later step and $A$ can be triggered only if
$\beta$ was true in some previous step; formally, for any trace $\sigma$
of the system, if $\beta$ is true along $\sigma$ at the time point $i$
then $\trigger{A}$ is true in $\sigma[j]$ for some $j\geq i$ (Completeness); if $\trigger{A}$ is
true in $\sigma[i]$, then $\beta$ must be true along $\sigma$ in some time
point between $0$ and $i$ (Correctness).\\

\begin{figure*}[t]
%\vspace{-2em}
\center
\begin{tabular}{|c|l|l|} \hline
$\beta$  & \pptracebeta \\ \hline
\edel{A}{\beta}{2} & \pptraceedel \\ \hline
\bdel{A}{\beta}{4} & \pptracebdel \\ \hline
\fdel{A}{\beta}    & \pptracefdel \\ \hline
\end{tabular}
\caption{Examples of alarm responses to the diagnosis condition $\beta$.}
\label{tbl-alarmconditions}
%\vspace{-1.5em}
\end{figure*}

Figure~\ref{tbl-alarmconditions} provides an example of admissible
responses for the various alarms to the occurrences of the same
diagnosis condition $\beta$; note how in the case of
\bdel{A}{\beta}{4} the alarm can be triggered at any point as long as
it is within the next 4 time-steps. Since $A$ is a state variable and
the diagnoser changes it only in response to synchronizations with the
plant, every rising and falling edge of the alarm in the figure
corresponds to an observation point.

\begin{figure}
\center
\begin{tabular}{|l|p{0.53\textwidth}|}
\hline
\textbf{Pattern} & \textbf{Description} \\ \hline
\edel{PSU1_{Exact_i}}{\beta_{PSU 1}}{i} &
Detect if the PSU 1 (Generator 1 + Battery 1) is broken, in order to switch
to secondary mode \\
\bdel{PSU1_{Bound}}{\beta_{PSU1}}{C} &
Detect if the PSU (Generator 1 + Battery 1) was broken within the
bound, in order to switch to secondary mode \\
\bdel{BS}{\beta_{BS}}{DC} & Detect if the whole Battery-Sensor block
is working incorrectly, in order to replace it\\
\fdel{Discharged}{\beta_{Depleted}} & Detect if any of the battery was
ever completely discharged\\
\hline
\end{tabular}
\caption{Example Specification for the BSS}
\label{fig:bss-spec}
\end{figure}

Figure~\ref{fig:bss-spec} contains a simple specification for our
running example. There are two types of PSU (Power Supply Unit) alarms
(that can be similarly defined for PSU 2). The first one defines
multiple alarms, each having a different delay $i$. Let us assume that
each battery has a capacity $C$ of 10, and that this provides us with
a delay of at most 10 time-units. We can instantiate 10 alarms one for
each $i \in [0,10]$. Ideally, we want to detect the exact moment in
which the PSU stop working. However, this might not be possible due to
non-diagnosability. Therefore, we define a weaker version of the alarm
($PSU1_{Bound}$), in which we say that within the time-bound provided
by the battery capacity ($C$) we want to know if the PSU stop
working. In Section~\ref{sec-fdi-validation} we will prove that one
alarm condition is weaker than the other. For most alarms, we specify
what recovery can be applied to address the problem. In this way, our
process of defining the alarms of interest is driven by the recovery
procedures available.  If there is no automated recovery for a given
situation, time-bounds might not be relevant anymore. Therefore, we
use alarms to collect information on the historical state of the
system (e.g., $Discharged$ alarm); notice, in fact, that
\textsc{FiniteDel} alarm have a permanent behavior, i.e., they can
never be turned off.

\subsection{Diagnosability}

Given an alarm condition, we need to know whether it is possible to
build a diagnoser for it. In fact, there is no reason in having a
specification that cannot be realized. This property is called
\emph{diagnosability} and was introduced in~\cite{Sampath95}.

In this section, we define the concept of diagnosability for the
different types of alarm conditions. We proceed by first giving the
definition of diagnosability in the traditional way (\`a la Sampath)
in terms of observationally equivalent traces w.r.t.\ the diagnosis
condition. Then, we prove that a plant $P$ is diagnosable iff there
exists a diagnoser that satisfies the specification.

\begin{defi}\label{def-edel-diagnosability}
Given a plant $P$ and a diagnosis condition $\beta$, we say that
$\edel{A}{\beta}{\delay}$ is diagnosable in $P$ iff for all $\sigma_1, i$
s.t.\ $\sigma_1,i\models\beta$ then $\ObsPoint(\sigma_1,i+\delay)$ and for
all $\sigma_2, j$, if $\ObsEquiv((\sigma_1,i+\delay),(\sigma_2,j+\delay))$, then
$\sigma_2,j\models\beta$.
\end{defi}

\noindent
Therefore, an exact-delay alarm condition is not diagnosable in $P$
iff either there is no synchronization after $\delay$ steps (note that this
is not possible in the synchronous case) or there exists a pair of
traces $\sigma_1$ and $\sigma_2$ such that for some $i,j\geq 0$,
$\sigma_1,i\models\beta$, $ObsEq((\sigma_1, i+\delay), (\sigma_2, j+\delay))$,
and $\sigma_2,j\not\models\beta$. We call such a pair a \emph{critical
  pair}.

\begin{defi}\label{def-bdel-diagnosability}
Given a plant $P$ and a diagnosis condition $\beta$, we say that
$\bdel{A}{\beta}{\delay}$ is diagnosable in $P$ iff forall $\sigma_1, i$
s.t.\ $\sigma_1,i\models\beta$ there exists $k$ s.t.\ $i \leq k \leq
i+\delay$, $\ObsPoint(\sigma_1,k)$ and for all $\sigma_2, l$, if
$\ObsEquiv((\sigma_1,k),(\sigma_2,l))$, then there exists $j$ s.t.\ $l-\delay
\leq j \leq l$ and $\sigma_2,j\models\beta$.
\end{defi}

Intuitively, $k,l$ denote points that are observationally equivalent
and $i,j$ denote the states where the condition occurred, and their
relation is such that $i$ and $j$ do not occur more than $\delay$ steps
away from each other.

This definition takes into account occurrences of $\beta$ that
happened before $i$. Indeed, we need to check occurrences
up to $\delay$ states before and after $i$. Consider the two traces
$\sigma_1 = apbqc$ and $\sigma_2 = aqbpc$, where $a,b,c$ are
observable events, and $\beta = p$. We can see that we can justify $p$
in $\sigma_1$ by looking at the occurrence of $p$ in $\sigma_2$ that
is in the future. However, we cannot justify the $p$ in $\sigma_2$ by
just looking in the future, but we need to look in the past.
\begin{defi}\label{def-fdel-diagnosability}
Given a plant $P$ and a diagnosis condition $\beta$, we say that
$\fdel{A}{\beta}$ is diagnosable in $P$ iff for all $\sigma_1, i$
s.t.\ $ \sigma_1, i \models \beta$ then there exist $k \ge i$
s.t.\ $\ObsPoint(\sigma_1,k)$ and for all $\sigma_2, l$ if
$\ObsEquiv((\sigma_1,k), (\sigma_2,l))$ then there exists $j \leq l$
$\sigma_2,j\models\beta$.
\end{defi}

Definition~\ref{def-bdel-diagnosability} is a generalization
of Sampath's definition of diagnosability:
\begin{defi}{(Diagnosability~\cite{Sampath95})}
\label{def:sampath}
Given a plant $P$ and a diagnosis condition $\beta$, we say that
$\beta$ is diagnosable in $P$ iff there exists $\delay$ s.t.\ for all
$\sigma, i$, $\sigma_2, l$, $k \ge i+\delay$ if $\sigma_1, i \models \beta$ and
$\Obs(\sigma_2^{l}) = \Obs(\sigma_1^{k})$ then there
exists $j \le l$ s.t.\ $\sigma_2, j \models \beta$.
\end{defi}

In \cite{Sampath95} (specifically in Section II.A), Sampath et
al.\ also assume that there are no cycles of unobservable events. This
means that there is a $d_u$ s.t.\ for all $\sigma, i$ s.t.\ $\sigma, i
\models \beta$ then there exists $k$ s.t. $0 \le k \le d_u$ and
$ObsPoint(\sigma, i+k)$.

\begin{thm}\label{th-classical-diagnosability}
Let $P$ be a plant such that there is no cycle of unobservable events,
and let $p$ be a propositional formula, then $p$ is diagnosable (as
defined in~\ref{def:sampath}) in $P$ iff there exists $\delay$ such that
$\bdel{A}{Op}{\delay}$ is diagnosable in $P$.
\end{thm}

\begin{proof}\
  \begin{itemize}[label=$\Rightarrow$)]
\item[$\Rightarrow$)] Assume that $p$ is diagnosable in $P$. Consider a
  trace $\sigma_1$ such that for some $i\geq 0$, $\sigma_1,i\models O
  p$. Then, for some $0\leq i'\leq i$, $\sigma_1,i'\models p$. By
  assumption, we know that there is a $\delay$ s.t.\ for all $k \ge i'+\delay$
  and any trace $\sigma_2$ and point $l$ such that
  $\Obs(\sigma^l_2)=\Obs(\sigma^k_1)$ then $\sigma_2,j'\models p$ for
  some $j'$, $j'\leq l$. Then $\sigma_2,j\models Op$ for all $j \ge
  j'$. Since this holds for any $k$ and $l$, it holds also for the $k$
  and $l$ that are observation points for $\sigma_1$ and
  $\sigma_2$. Let $\delay'=\delay+n_u$. Then there exists $k'<\delay'$ such that
  $\ObsPoint(\sigma_1,i+k')$ and for all trace $\sigma_2$ and point
  $l$ such that $\ObsEquiv((\sigma_1,k'),(\sigma_2,l))$ then
  $\sigma_2,j'\models p$ for some $j'$, $j'\leq l$.  We can conclude
  that $\bdel{A}{Op}{\delay'}$ is diagnosable in $P$.

\item[$\Leftarrow$)] Assume that $\bdel{A}{Op}{\delay}$ is diagnosable in
  $P$. Consider a trace $\sigma_1$ such that for some $i\geq 0$
  $\sigma_1,i\models p$. Then $\sigma_1,i\models O p$. By assumption,
  there exists $k$, $i \leq k\leq i+\delay$ such that
  $\ObsPoint(\sigma_1,k)$ and, for any trace $\sigma_2$ and point $l$
  such that $ObsEq((\sigma_1, k), (\sigma_2, l))$ then
  $\sigma_2,j\models Op$ for some $l-\delay\leq j\leq l$. Let us consider
  $\sigma_2$ and $l$ such that $\Obs(\sigma_2^{l}) =
  \Obs(\sigma_1^{k})$. Then for some $l'\leq l$ we have that
  $\ObsPoint(\sigma_2,l')$ and therefore $ObsEq((\sigma_1, k),
  (\sigma_2, l'))$. Then $\sigma_2,j\models Op$ for some $l-\delay\leq
  j\leq l$. Thus $\sigma_2, j' \models p$ for some $j' \le j$ and $P$
  is diagnosable.\qedhere
  \end{itemize}
\end{proof}

\noindent The following theorem shows that if a component satisfies the
diagnoser specification then the monitored plant must be diagnosable
for that specification. In Section \ref{sec-fdi-synthesis} on
synthesis we will show also the converse, i.e., if the specification
is diagnosable then a diagnoser exists.

\begin{thm}\label{th-diagnosable}
Let $D$ be a diagnoser for $P$. If $D$ satisfies an alarm
condition then the alarm condition is diagnosable in $P$.
\end{thm}

\begin{proof}
By contradiction, suppose $\edel{A}{\beta}{\delay}$ is not diagnosable in
$P$. Then either there exists a trace $\sigma_1$ with
$\sigma_1,i\models\beta$ for some $i$ such that
$\ObsPoint(\sigma_1,j)$ is false for all $j\geq i$ or there exists a
critical pair. In the first case, $A$ is not triggered and the
diagnoser is not complete.  Suppose there exists a critical pair of
traces $\sigma_1$ and $\sigma_2$, i.e., for some $i,j\geq 0$
$\sigma_1,i\models\beta$, $ObsPoint(\sigma_1,i+\delay)$,
$\ObsEquiv((\sigma_1,i+\delay),(\sigma_2,j+\delay))$, and
$\sigma_2,j\not\models\beta$. Since $D$ is deterministic,
$D(\sigma_1)$ and $D(\sigma_2)$ have a common prefix compatible with
$obs(\sigma^{i+\delay}_1)=obs(\sigma^{j+\delay}_2)$. If the diagnoser is
complete then $A$ is triggered in $D(\sigma_1)\otimes\sigma_1$ at
position $i+\delay$, and so also in $D(\sigma_2)\otimes\sigma_2$ at
position $j+\delay$, but in this way the diagnoser is not correct, which is
a contradiction. If the diagnoser is correct, then $A$ is not
triggered in $D(\sigma_2)\otimes\sigma_2$ at position $j+\delay$, but so
neither in $D(\sigma_1)\otimes\sigma_1$ at position $i+\delay$, but in this
way the diagnoser is not complete, which is a contradiction.

Similarly, for $\fdel{A}{\beta}$ and $\bdel{A}{\beta}{\delay}$.
\end{proof}

\noindent The definition above of diagnosability might be stronger
than necessary, since diagnosability is defined as a global property
of the plant. Imagine the situation in which there is a critical pair
and after removing this critical pair from the possible executions of
the system, our system becomes diagnosable. This suggests that the
system was ``almost'' diagnosable, and an ideal diagnoser would be
able to perform a correct diagnosis in all the cases except one (i.e.,
the one represented by the critical pair). To capture this idea, we
redefine the problem of diagnosability from a global property
expressed on the plant, to a local property expressed on points of
single traces.

\begin{defi}\label{def-edel-local-diagnosability}
Given a plant $P$, a diagnosis condition $\beta$ and a trace
$\sigma_1$ such that for some $i\geq 0$ $\sigma_1,i\models\beta$, we
say that $\edel{A}{\beta}{\delay}$ is \emph{\dlocally\ diagnosable in}
$\mktuple{\sigma_1,i}$ iff $\ObsPoint(\sigma_1,i+\delay)$ and for any trace
$\sigma_2$, for all $j \ge 0$ such that
$\ObsEquiv((\sigma_1,i+\delay),(\sigma_2,j+\delay))$,
$\sigma_2,j\models\beta$.
\end{defi}

\begin{defi}\label{def-bdel-local-diagnosability}
Given a plant $P$, a diagnosis condition $\beta$, and a trace
$\sigma_1$ such that for some $i\geq 0$ $\sigma_1,i\models\beta$, we
say that $\bdel{A}{\beta}{\delay}$ is \emph{\dlocally\ diagnosable in}
$\mktuple{\sigma_1,i}$ iff there exists $k$ s.t.\ $i\leq k \leq i+\delay$,
$\ObsPoint(\sigma_1,k)$, and for any $\sigma_2, l$ if
$\ObsEquiv((\sigma_1,k),(\sigma_2,l))$, then there exists $j$
s.t.\ $l-\delay \leq k \leq l$ and $\sigma_2,j\models\beta$.
\end{defi}

\begin{defi}
Given a plant $P$, a diagnosis condition $\beta$, and a trace
$\sigma_1$ such that for some $i\geq 0$, $\sigma_1, i \models \beta$,
we say that $\fdel{A}{\beta}$ is \emph{\dlocally\ diagnosable in}
$\mktuple{\sigma_1,i}$ iff there exists $k \ge i$
s.t.\ $\ObsPoint(\sigma_1,k)$ and for all $\sigma_2, l$ if
$\ObsEquiv((\sigma_1,k),(\sigma_2,l))$, then there exists $j \leq l$
and $\sigma_2, j \models \beta$.
\end{defi}

A specification that is \dlocally\ diagnosable in a plant along all
points of all traces is diagnosable in the classical sense, and we say
it is \emph{\dglobally} diagnosable. The concept of
\dlocal\ diagnosability does not impose any specific behavior to the
diagnoser. However, it is an important concept that allows us to
better characterize and understand the specification and the system.

\subsection{Maximality}

As shown in Figure~\ref{tbl-alarmconditions}, bounded- and
finite-delay alarms are correct if they are raised within the
valid bound. However, there are several possible variations
of the same alarm in which the alarm is active in different instants
or for different periods. We address this problem by introducing the
concept of \emph{maximality}.  Intuitively, a maximal diagnoser is
required to raise the alarms as soon as possible and as long as
possible (without violating the correctness condition).

\begin{defi}\label{def-maximality}
$D$ is a maximal diagnoser for an alarm condition with alarm $A$ in
  $P$ iff for every trace $\sigma_{P}$ of $P$, $D(\sigma_P)$ contains
  the maximum number of observable points $i$ such that
  $D(\sigma_P),i\models A$; that is, if
  $D(\sigma_P),i\not\models A$, then there
  does not exist another correct diagnoser $D'$ of $P$ such that
  $D'(\sigma_P),i\models A$.
\end{defi}

\section{Formal Specification}
\label{sec-epistemic}

In this section, we present the Alarm Specification Language with
Epistemic operators (\kasl). This language allows designers to
define requirements on the FDI alarms including aspects such as
delays, diagnosability and maximality.

Diagnosis conditions and alarm conditions are formalized using LTL
with past operators. The
definitions of \dlocal\ diagnosability and maximality, however, cannot
be captured by using a formalization based on LTL. To capture these
two concepts, we rely on temporal epistemic logic.
The intuition is that this logic enables us to reason on set of
observationally equivalent traces instead that on single traces (like
in LTL). We show how this logic can be used to specify diagnosability, define
requirements for non-diagnosable cases and express the concept of
maximality.

\subsection{Diagnosis and Alarm Conditions as LTL Properties}
Let $\props$ be a set of propositions representing either faults,
events or elementary conditions for the diagnosis.
The set $\diags_\props$ of \emph{diagnosis conditions} over $\props$
is any formula $\beta$ built with the following rule:
\[ \beta ::= p \mid \beta \land \beta \mid \lnot \beta \mid O \beta \mid
Y \beta \]
\noindent
with $p \in \props$.

We provide the LTL characterization of the \emph{Alarm Specification
  Language} (ASL) in Figure~\ref{fig:asl}. On the left column we
provide the name of the alarm condition (as defined in the previous
section), and on the right column we provide the associated LTL
formalization encoding the concepts of correctness and
completeness. \emph{Correctness}, the first conjunct, intuitively says
that whenever the diagnoser raises an alarm, then the fault must have
occurred. \emph{Completeness}, the second conjunct, intuitively
encodes that whenever the fault occurs, the alarm will be raised.  In
the following, for simplicity, we abuse notation and indicate with
$\varphi$ both the alarm condition and the associated LTL; for an
alarm condition $\varphi$, we denote by $A_\varphi$ the associated
alarm variable $A$, and with $\tau(\varphi)$ the following formulas:
\begin{itemize}
\item[] $\tau(\varphi)=Y^\delay\beta$ for $\varphi=\edel{A}{\beta}{\delay}$;
\item[] $\tau(\varphi)=O^{\le \delay}\beta$ for $\varphi=\bdel{A}{\beta}{\delay}$;
\item[] $\tau(\varphi)=O\beta$ for $\varphi=\fdel{A}{\beta}$.
\end{itemize}
\noindent When clear from the context, we use just $A$ and $\tau$
instead of $A_\varphi$ and $\tau(\varphi)$, respectively.

\begin{figure}[ht]
\center
\begin{tabular}{|l|l|}
\hline
\textbf{Alarm Condition} & \textbf{LTL Formulation}  \\ \hline
$\edel{A}{\beta}{\delay}$ &
  $\correctness{G(\trigger{A} \rightarrow Y^\delay \beta)} \wedge
   \completeness{G(\beta \rightarrow X^\delay \trigger{A})}$ \\ \hline
$\bdel{A}{\beta}{\delay}$ &
  $\correctness{G(\trigger{A} \rightarrow {O^{\le \delay} \beta})} \wedge
   \completeness{G(\beta \rightarrow F^{\le \delay} \trigger{A})}$ \\ \hline
$\fdel{A}{\beta}$ &
  $\correctness{G(\trigger{A} \rightarrow O \beta)} \wedge
   \completeness{G(\beta \rightarrow F \trigger{A})}$ \\ \hline
\end{tabular}
\caption{Alarm conditions as LTL (ASL):
\correctness{Correctness} and \completeness{Completeness}}
\label{fig:asl}
%\vspace{-1em}
\end{figure}
\subsection{Diagnosability as Epistemic Property}
\label{sec-epistemic-diag}
\begin{figure}[t]
      \begin{tabular}{|l|l|l|} \hline
            \textbf{Alarm Condition} & \textbf{Diagnosability} & \textbf{Maximality}
            \\ \hline
            \edel{A}{\beta}{\delay} &
            $\diagnosability{G ( \beta \rightarrow X^\delay \strong{K Y^\delay \beta} )}$
            &
            $\maximality{ G (\strong{K Y^\delay \beta} \rightarrow \trigger{A}) }$     \\ \hline
            \bdel{A}{\beta}{\delay} &
            $\diagnosability{G ( \beta \rightarrow F^{\le \delay} \strong{K O^{\le \delay} \beta})} $
            &
            $\maximality{ G (\strong{K O^{\le \delay} \beta} \rightarrow \trigger{A}) } $ \\ \hline
            \fdel{A}{\beta}    &
            $\diagnosability{G ( \beta \rightarrow F \strong{K O \beta} )}$
            &
            $\maximality{G (\strong{K O \beta} \rightarrow \trigger{A} )}$  \\ \hline
      \end{tabular}

    \caption{\diagnosability{Diagnosability} and \maximality{Maximality}.}
    \label{fig:diagnosability-maximality}
\end{figure}

\begin{figure}[t]
  \begin{center}
    \maximaltracesk
    \caption{Example of Maximal and Non-Maximal traces}
    \label{fig:maximal-trace}
  \end{center}
\end{figure}

We can write the diagnosability test for the different alarm
conditions directly as epistemic properties. The general formulation
is presented on the left column of
Figure~\ref{fig:diagnosability-maximality}.
In order to test for system diagnosability, we will check whether the
formula holds for all traces of the system; while to check for
\dlocal\ diagnosability we will check whether the formula holds for
single points in a trace.  For example, the diagnosability test for
\edel{A}{\beta}{\delay} says that it is always the case that whenever
$\beta$ occurs, exactly $\delay$ steps afterwards, the diagnoser
\emph{knows} $\beta$ occurred $\delay$ steps earlier. Since $K$ is
defined on observationally equivalent traces, the only way to falsify
the formula would be to have a trace in which $\beta$ occurs, and
another one (observationally equivalent at least for the next $\delay$
steps) in which $\beta$ did not occur; but this is in contradiction
with the definition of diagnosability
(Definition~\ref{def-edel-diagnosability}).

\subsection{Maximality as Epistemic Property}

The property of maximality says that the diagnoser will raise the
alarm as soon as it is possible to know the diagnosis condition, and
the alarm will stay up as long as possible.  The property $ \strong{K
\tau} \rightarrow \trigger{A} $ encodes this behavior:
\begin{thm}\label{th-epistemic-maximality}
$D$ is maximal for $\varphi$ in $P$ iff $D\otimes P\models
  G(\strong{K\tau}\rightarrow \trigger{A})$.
\end{thm}
\begin{proof}
$\Rightarrow$) Suppose $D$ is maximal and by contradiction $D\otimes
  P\not\models G(\strong{K\tau}\rightarrow \trigger{A})$. Thus, there
  exists a trace $\sigma_P$ of $P$ and $i\geq 0$ such that
  $D(\sigma_P)\times \sigma_P,i\models(\strong{K\tau}\wedge\neg
  \trigger{A})$ (where $D(\sigma_P)$ is the diagnoser trace matching
  $\sigma_P$ as defined in Definition~\ref{def-matchingtrace}). By
  Definition~\ref{def-observed} of $\trigger{\cdot}$, $i$ is an
  observation point.
  Let $i$ be the $j$-th observation point of
  $\sigma_P$. Consider $D'$ obtained by $D(\sigma_p)$ converting the
  trace into a transition system using a sink state so that $D'$ is
  deterministic and setting $\trigger{A}$ to true only in the state
  $D(\sigma_P)[j]$ (thus triggering $A$ in $j$ and setting it to false
  at the next observation point). For every trace $\sigma'_P$ of $P$
  matching with $D'(\sigma_P$), $\Obs(\sigma'_P)=\Obs(\sigma_P)$, and
  thus $\sigma'_P,i\models \tau$ (since $D(\sigma_P)\times
  \sigma_P,i\models(\strong{K\tau}$). Therefore $D'\models
  G(\trigger{A}\rightarrow \tau)$ contradicting the hypothesis.

$\Leftarrow$) Suppose $D\otimes P\models G(\strong{K\tau}\rightarrow
  \trigger{A})$ and by contradiction $D$ is not maximal for $\varphi$
  in $P$.  Then there exists a trace $\sigma_P$ of $P$ such that
  $D(\sigma_P),i\not\models \trigger{A}$ and there exists another
  diagnoser $D'$ of $P$ such that $D'(\sigma_P),i\models \trigger{A}$
  and $D'\otimes P\models G(\trigger{A}\rightarrow \tau)$. Then, for
  some $j$, $D(\sigma_P)\otimes \sigma_P,j\not\models \trigger{A}$,
  $D'(\sigma_P)\otimes\sigma_P,j\models \trigger{A}$, and so
  $D(\sigma_P)\otimes\sigma_P,j\not\models \strong{K\tau}$ and
  $\sigma_P,j\models \tau$. Then there exists another trace $\sigma'_P$
  of $P$ and $j'$ such that $\ObsEquiv((\sigma'_P,j'),(\sigma_P,j))$
  and $\sigma'_P,j'\not\models \tau$. Since $D'$ is deterministic,
  $D'(\sigma'_P)$ and $D'(\sigma_P)$ are equal up to position $i$, and
  so $D'\otimes P\not\models G(\trigger{A}\rightarrow \tau)$
  contradicting the hypothesis.
\end{proof}

\noindent
Whenever the diagnoser knows that $\tau$ is satisfied, it
will raise the alarm. An example of maximal and non-maximal alarm is
given in Figure~\ref{fig:maximal-trace}. Note that according to our
definition, the set of maximal alarms is a subset of the non-maximal
ones.

A property related to Maximality is the capability of the diagnoser to
justify the raising of the alarm. This property is guaranteed by
construction by any correct diagnoser, as shown in the following
theorem.
\begin{thm}\label{thm-alarm-implies-knowledge}
Given a diagnoser $D$ and a plant $P$, for each alarm $A$ of
$D$, with temporal condition $\tau$, if $D$ is correct for
$A$ it holds that:
 \[ D \otimes P \models G(\trigger{A} \rightarrow \strong{K \tau}) \]
Thus, whenever the diagnoser raises an alarm, it knows that
the diagnosis condition has occurred.
\end{thm}

\begin{proof}
We assume by contradiction that the $G(\trigger{A} \rightarrow
\strong{K \tau})$ is not satisfied. Therefore, there exist $\sigma$
and $i$ such that $D(\sigma)\otimes\sigma,i\models\trigger{A} \land
\lnot \strong{K\tau}$ (where $D(\sigma_P)$ is the diagnoser trace
matching $\sigma_P$ as defined in Definition~\ref{def-matchingtrace}),
which is equivalent to $\trigger{A} \land \lnot K \tau$ (by
Definition~\ref{def-observed} of $\trigger{\cdot}$). Thus,
$\sigma,i\models\tau$ by correctness of $D$.
In order for the $\lnot K \tau$ to hold, we need another trace
$\sigma'$ and $j$ s.t.\ $\ObsEquiv((\sigma, i),(\sigma', j))$ and
$\sigma', j \models \lnot \tau$.
By definition, the diagnoser is deterministic, thus we know that for
$\sigma,\sigma'$ at points $i,j$ we will have the same value of $A$.
Therefore, $D(\sigma')\otimes\sigma', j \models \trigger{A} \land
\lnot \tau$ so that $D$ is not correct, thus reaching a
contradiction.
\end{proof}

\subsection{\texorpdfstring{\kasl}{ASLk} Specifications}
\begin{figure*}[t]
\center
\scalebox{0.8}{%
\begin{tabular}{|l|l||l||l|} \hline
& Template &  $\maximality{Maximality} = False$& $\maximality{Maximality} = True$  \\
\hline \hline
\parbox[t]{5mm}{\multirow{3}{*}{\rotatebox[origin=c]{90}{$\diagnosability{Diag} = \dGlobal$ }}}
& \multirow{2}{*}{\textsc{ExactDel}}
& %\cellcolor[gray]{.9}
  $\correctness{G (\trigger{A} \rightarrow {Y^{\delay} \beta})} \wedge
  \completeness{G (\beta \rightarrow X^\delay \trigger{A})}$
& $\correctness{G (\trigger{A} \rightarrow {Y^{\delay} \beta})} \wedge
  \completeness{G (\beta \rightarrow X^\delay \trigger{A})} \ \wedge$ \\
& & %\cellcolor[gray]{.9}
& $\maximality{G (\strong{K Y^\delay \beta} \rightarrow \trigger{A})} $\\ \cline{2-4}
& \multirow{2}{*}{\textsc{BoundDel}}
& %\cellcolor[gray]{.9}
  $\correctness{G (\trigger{A} \rightarrow {O^{\le \delay} \beta})} \wedge
   \completeness{G (\beta \rightarrow F^{\le \delay} \trigger{A})} $
& $\correctness{G (\trigger{A} \rightarrow {O^{\le \delay} \beta})} \wedge
  \completeness{G (\beta \rightarrow F^{\le \delay} \trigger{A})} \ \wedge$ \\
& & %\cellcolor[gray]{.9}
& $ \maximality{G (\strong{K O^{\le \delay} \beta} \rightarrow \trigger{A})}$
  \\ \cline{2-4}
& \multirow{2}{*}{\textsc{FiniteDel}}
& %\cellcolor[gray]{.9}
  $\correctness{G (\trigger{A} \rightarrow {O \beta})} \wedge
   \completeness{G (\beta \rightarrow F \trigger{A})}$
& $\correctness{G (\trigger{A} \rightarrow {O \beta})} \wedge
   \completeness{G (\beta \rightarrow F \trigger{A})} \ \wedge$ \\
& & %\cellcolor[gray]{.9}
&$\maximality{G (\strong{K O \beta} \rightarrow \trigger{A})} $ \\ \cline{2-4}
\hline \hline
\parbox[t]{5mm}{\multirow{3}{*}{\rotatebox[origin=c]{90}{ $\diagnosability{Diag} =
      \dLocal$ \ \ \ \  }}}
& \multirow{2}{*}{\textsc{ExactDel}}
& $\correctness{G (\trigger{A} \rightarrow {Y^{\delay} \beta})} \ \wedge$
& $\correctness{G(\trigger{A} \rightarrow {Y^{\delay} \beta})} \ \wedge$ \\
& & $\completeness{G( \diagnosability{(\beta \rightarrow X^\delay \strong{K Y^\delay \beta} )}
  \rightarrow (\beta \rightarrow X^\delay \trigger{A}))}$
& $\completeness{G( \diagnosability{(\beta \rightarrow X^\delay \strong{K Y^\delay \beta} )}
  \rightarrow (\beta \rightarrow X^\delay \trigger{A}))} \ \wedge$ \\
& & & $\maximality{G (\strong{K Y^\delay \beta} \rightarrow \trigger{A})}$ \\
\cline{2-4}
& \multirow{2}{*}{\textsc{BoundDel}}
& $\correctness{G (\trigger{A} \rightarrow {O^{\le \delay} \beta})}\ \wedge$
& $\correctness{G (\trigger{A} \rightarrow {O^{\le \delay} \beta})}\ \wedge$ \\
& & $\completeness{G( \diagnosability{(\beta \rightarrow F^{\le \delay} \strong{K O^{\le \delay} \beta} )}
    \rightarrow (\beta \rightarrow F^{\le \delay} \trigger{A}))}$
& $\completeness{G( \diagnosability{(\beta \rightarrow F^{\le \delay} \strong{K O^{\le \delay} \beta} )}
 \rightarrow (\beta \rightarrow F^{\le \delay} \trigger{A}))} \ \wedge $ \\
& & & $\maximality{G(\strong{K O^{\le \delay} \beta} \rightarrow \trigger{A})}$ \\
\cline{2-4}
& \multirow{2}{*}{\textsc{FiniteDel}}
& $\correctness{G (\trigger{A} \rightarrow {O \beta})}\ \wedge $
& $\correctness{G (\trigger{A} \rightarrow {O \beta})}\ \wedge $ \\
& & $\completeness{G( \diagnosability{(\beta \rightarrow F \strong{K O \beta} )}
  \rightarrow (\beta \rightarrow F \trigger{A}))} $
& $\completeness{G( \diagnosability{(\beta \rightarrow F \strong{K O \beta} )}
  \rightarrow (\beta \rightarrow F \trigger{A}))} \ \wedge$ \\
& & & $\maximality{G(\strong{K O \beta} \rightarrow \trigger{A})} $ \\
\cline{2-4}
\hline
\end{tabular}}
\caption{\kasl specification patterns among the four dimensions:
  \diagnosability{Diagnosability}, \maximality{Maximality},
  \completeness{Completeness} and \correctness{Correctness}.}
\label{fig:kasl}
\end{figure*}

The formalization of \kasl (Figure~\ref{fig:kasl}) is obtained by
extending ASL (Figure~\ref{fig:asl}) with the concepts of maximality
and diagnosability, defined as epistemic properties. When
\emph{maximality} is required we add a third conjunct following
Theorem~\ref{th-epistemic-maximality}. When $Diag = \dLocal$ instead,
we precondition the completeness to the \dlocal\ diagnosability (as
defined in Figure~\ref{fig:diagnosability-maximality}); this means
that the diagnoser will raise an alarm whenever the diagnosis
condition is satisfied and the diagnoser is able to know it.

Several simplifications are possible. For example, in the case
$Diag=\dLocal$, we do not need to verify the completeness due to the
following result:
\begin{thm}\label{thm-max-local-for-completeness}
 Given a diagnoser $D$ for a plant $P$ and a \dlocally\ diagnosable
 alarm condition $\varphi$, if $D$ is maximal for $\varphi$, then
 $D$ is complete.
\end{thm}
\begin{proof}
(\textsc{ExactDel})
For all $\sigma, i$
if $\sigma,i\models (\beta\rightarrow X^{\delay} \strong{K Y^\delay \beta})$,
then by using the maximality assumption, we know that
$\sigma,i\models (\beta\rightarrow X^\delay \trigger{A})$;
thus, $\sigma,i\models (\beta\rightarrow X^{\delay} \strong{K Y^{\delay}\beta}) \rightarrow
(\beta\rightarrow X^\delay \trigger{A})$.
Similarly we can prove \textsc{BoundDel} and \textsc{FiniteDel}.
%%%% MG: The following are the actual proofs. But I think it is redundant to
%%%%     Include them in-text.
%% \item[BoundDel]
%% For all $\sigma, i$
%% if $\sigma,i\models (\beta\rightarrow F^{\le n} K O^{\le n} \beta)$,
%% then by using the maximality assumption, we know that
%% $\sigma,i\models (\beta\rightarrow F^{\le n} A)$;
%% thus, $\sigma,i\models (\beta\rightarrow F^{\le n} K O^{\le n}\beta)\rightarrow
%%   (\beta\rightarrow F^{\le n} A)$.
%% %
%% \item[FiniteDel]
%% For all $\sigma, i$
%% if $\sigma,i\models (\beta\rightarrow F K O \beta)$,
%% then by using the maximality assumption, we know that
%% $\sigma,i\models (\beta\rightarrow F A)$;
%% thus, $\sigma,i\models (\beta\rightarrow F K O \beta)\rightarrow
%% (\beta\rightarrow F A)$.
\end{proof}

As a corollary of Theorem~\ref{thm-max-local-for-completeness}, the
same can be applied also for \dglobally\ diagnosable alarm conditions
if $P$ is diagnosable, since \dglobal\ diagnosability implies
\dlocal\ diagnosability:
\begin{thm}\label{thm-max-global-for-completeness}
 Given an alarm condition for the \dglobally\ diagnosable case, and a
 diagnoser $D$ for a plant $P$, if $D$ is maximal for $\varphi$ and
 $\varphi$ is diagnosable in $P$ then $D$ is complete.
\end{thm}

\begin{proof}
The theorem follows directly from Theorem~\ref{thm-max-local-for-completeness} and the fact that if $D$
is complete for a trace diagnosable alarm condition that is system
diagnosable, then $D$ is also complete for the corresponding system
diagnosable alarm condition.
\end{proof}
\noindent This Theorem is interesting because it tells us that if a
specification that was required to be \dglobally\ diagnosable is
indeed \dglobally\ diagnosable, then we can just check whether the
diagnoser is maximal and avoid performing the completeness test.

\begin{figure*}[ht]
\scalebox{0.9}{%
\begin{tabular}{|l|l||l||l|} \hline
& Template &  $\maximality{Maximality} = False$& $\maximality{Maximality} = True$  \\
\hline \hline
\parbox[t]{5mm}{\multirow{3}{*}{\rotatebox[origin=c]{90}{$\diagnosability{Diag} = \dGlobal$ }}}
& \multirow{2}{*}{\textsc{ExactDel}}
& %\cellcolor[gray]{.9}
  $\correctness{G (\trigger{A} \rightarrow {Y^{\delay} \beta})} \wedge
  \completeness{G (\beta \rightarrow X^\delay \trigger{A})}$
& $\correctness{G (\trigger{A} \rightarrow {Y^{\delay} \beta})} \wedge
  \completeness{G (\beta \rightarrow X^\delay \trigger{A})}$ \\
& & %\cellcolor[gray]{.9}
& $\maximality{G (\strong{K Y^{\delay} \beta} \rightarrow {A})}$ \\ \cline{2-4}
& \multirow{2}{*}{\textsc{BoundDel}}
& %\cellcolor[gray]{.9}
  $\correctness{G (\trigger{A} \rightarrow {O^{\le \delay} \beta})} \wedge
   \completeness{G (\beta \rightarrow F^{\le \delay} \trigger{A})} $
& $\correctness{G (\trigger{A} \rightarrow {O^{\le \delay} \beta})} \wedge
  \completeness{G (\beta \rightarrow F^{\le \delay} \trigger{A})} \ \wedge$ \\
& & %\cellcolor[gray]{.9}
& $ \maximality{G (\strong{K O^{\le \delay} \beta} \rightarrow {A})}$
  \\ \cline{2-4}
& \multirow{2}{*}{\textsc{FiniteDel}}
& %\cellcolor[gray]{.9}
  $\correctness{G (\trigger{A} \rightarrow {O \beta})} \wedge
   \completeness{G (\beta \rightarrow F \trigger{A})}$
& $\correctness{G (\trigger{A} \rightarrow {O \beta})} \wedge
   \completeness{G (\beta \rightarrow F \trigger{A})} \ \wedge$ \\
& & %\cellcolor[gray]{.9}
&$\maximality{G (\strong{K O \beta} \rightarrow {A})} $ \\ \cline{2-4}
\hline \hline
\parbox[t]{5mm}{\multirow{3}{*}{\rotatebox[origin=c]{90}{ $\diagnosability{Diag} =
      \dLocal$ \ \ \ \  }}}
& \multirow{2}{*}{\textsc{ExactDel}}
& $\correctness{G (\trigger{A} \rightarrow {Y^{\delay} \beta})} \ \wedge$
& $\correctness{G(\trigger{A} \rightarrow {Y^{\delay} \beta})} \ \wedge$ \\
& & $\maximality{G (\strong{K Y^\delay \beta} \rightarrow {A})}$
& $\maximality{G (\strong{K Y^\delay \beta} \rightarrow {A})}$ \\
\cline{2-4}
& \multirow{2}{*}{\textsc{BoundDel}}
& $\correctness{G (\trigger{A} \rightarrow {O^{\le \delay} \beta})}\ \wedge$
& $\correctness{G (\trigger{A} \rightarrow {O^{\le \delay} \beta})}\ \wedge$ \\
& & $\completeness{G(( \beta \land \diagnosability{F^{\le \delay} \strong{K O^{\le \delay} \beta}} )
    \rightarrow F^{\le \delay} \trigger{A})}$
& $\maximality{G(\strong{K O^{\le \delay} \beta} \rightarrow {A})}$ \\
\cline{2-4}
& \multirow{2}{*}{\textsc{FiniteDel}}
& $\correctness{G (\trigger{A} \rightarrow {O \beta})}\ \wedge $
& $\correctness{G (\trigger{A} \rightarrow {O \beta})}\ \wedge $ \\
& & $\completeness{G(( \beta \land \diagnosability{F \strong{K O \beta}} )
  \rightarrow F \trigger{A})} $
& $\maximality{G(\strong{K O \beta} \rightarrow {A})} $ \\
\cline{2-4}
\hline
\end{tabular}}
\caption{\kasl with simplified patterns for $\diagnosability{Diag}=Trace$}
\label{fig:kasl-simplified}
\end{figure*}

\begin{thm}\label{thm-exactdel-maximality}
For all \dlocally\ diagnosable and non-maximal \textsc{ExactDel}
specifications, completeness can be replaced by maximality. Formally,
for all $\sigma$, $\sigma\models G( (\beta \rightarrow X^\delay
\strong{K Y^\delay \beta} ) \rightarrow (\beta \rightarrow X^\delay
\trigger{A}))$ iff $\sigma\models G (\strong{K Y^\delay \beta}
\rightarrow \trigger{A})$
\end{thm}
\begin{proof}
\begin{align*}
\sigma,i \models& ((\beta \rightarrow X^\delay \strong{K Y^\delay \beta} ) \rightarrow
(\beta \rightarrow X^\delay \trigger{A})) &\text{ iff }\\
\sigma,i \models & ( (\beta \land X^\delay \strong{K Y^\delay \beta} ) \rightarrow X^\delay \trigger{A}
) &\text{ iff } \\
\sigma,i+\delay \models& ( (Y^\delay \beta \land \strong{K Y^\delay \beta} ) \rightarrow \trigger{A}
) & \text{ iff } \\
\sigma,i+\delay \models& ( (\strong{Y^\delay \beta \land K Y^\delay \beta} ) \rightarrow \trigger{A}
) & \text{ iff } \\
\sigma,i+\delay \models& ( \strong{K Y^\delay \beta} \rightarrow \trigger{A}) &
\end{align*}
\noindent
Therefore, we can conclude that for all $i$, $\sigma,i \models ((\beta
\rightarrow X^\delay \strong{K Y^\delay \beta} ) \rightarrow (\beta \rightarrow
X^\delay \trigger{A}))$ iff for all $j\geq \delay$, $\sigma,j \models (
\strong{K Y^\delay \beta} \rightarrow \trigger{A})$. We conclude noting
that for $j<\delay$, $Y^\delay\beta$ is false and therefore $\sigma,j\models (
\strong{K Y^\delay \beta} \rightarrow \trigger{A})$.
\end{proof}

After applying the simplifications specified in
Theorem~\ref{thm-max-local-for-completeness} and
Theorem~\ref{thm-exactdel-maximality} and the equivalence
$\strong{\phi}\rightarrow\strong{\psi}\equiv\strong{\phi}\rightarrow\psi$,
we obtain the table in Figure~\ref{fig:kasl-simplified}, where the
patterns in the lower half ($Diag=Trace$) have been
simplified.

An \kasl specification is built by instantiating the patterns defined
in Figure~\ref{fig:kasl}. For example, we would write
$\kedel{A}{\beta}{\delay}{\dLocal}{True}$ for an exact-delay alarm $A$ for
$\beta$ with delay $\delay$, that satisfies the \dlocal\ diagnosability
property and is maximal. An introductory example on the usage of \kasl
for the specification of a diagnoser is provided in~\cite{DX2013}.
Figure~\ref{fig:bss-kasl-spec} shows how we extend the specification
for the BSS by introducing requirements on the diagnosability and
maximality of alarms. In particular, all the alarms that we defined
are not \dglobally\ diagnosable. Therefore, we need to weaken the
requirements and make them trace-diagnosable. The patterns are then
converted into temporal epistemic formulae as shown in
Figure~\ref{fig:bss-kasl-spec-klone}.
\begin{figure}[ht]
\center
\begin{tabular}{|l|p{0.53\textwidth}|}
\hline
%Pattern \\ \hline
\kedel{PSU1_{Exact_i}}{\beta_{PSU1}}{i}{Trace}{True} \\
\kbdel{PSU1_{Bound}}{\beta_{PSU1}}{C}{Trace}{True} \\
\kbdel{BS}{\beta_{BS}}{DC}{Trace}{True} \\
\kfdel{Discharged}{\beta_{Depleted}}{Trace}{False}\\
\kfdel{B1Leak}{\beta_{Battery1}}{System}{True}\\
\hline
\end{tabular}
\caption{\kasl Specification for the BSS}
\label{fig:bss-kasl-spec}
\end{figure}
\begin{figure}[ht]
\resizebox{\textwidth}{!}{
\begin{tabular}{|l|l|}
\hline
\textbf{Alarm} & \textbf{Formula} \\ \hline
$PSU1_{Exact_i}$ &
$ \correctness{G (\trigger{PSU1_{Exact_i}} \rightarrow Y^{i} \beta_{PSU 1})} \land
  \maximality{G (\strong{K Y^{i} \beta_{PSU 1}} \rightarrow \trigger{PSU1_{Exact_i}})}$
\\
$PSU1_{Bound}$ &
$ \correctness{G (\trigger{PSU1_{Bound}} \rightarrow O^{\le C} \beta_{PSU 1})} \land
  \maximality{G (\strong{K O^{\le C} \beta_{PSU 1}} \rightarrow \trigger{PSU1_{Bound}})}$
\\
$BS$ &
$ \correctness{G (\trigger{BS} \rightarrow O^{\le DC} \beta_{BS})} \land
  \maximality{G (\strong{K O^{\le DC} \beta_{BS}} \rightarrow \trigger{BS})}$
 \\
$Discharged$ &
$\correctness{G(\trigger{Discharged} \rightarrow O \beta_{Deplated})} \land
 \completeness{G((\beta_{Deplated} \land
   \diagnosability{F \strong{K O \beta_{Deplated}}})
\rightarrow F \trigger{Discharged})}$
\\
$B1Leak$ &
$\correctness{G( \trigger{B1Leak} \rightarrow O \beta_{Battery 1})} \land
 \completeness{G( \beta_{Battery 1} \rightarrow F \trigger{B1Leak})} \land
 \maximality{G( \strong{K O \beta_{Battery 1}} \rightarrow \trigger{B1Leak})}$
\\
\hline
\end{tabular}}
\caption{\klone\ translation of \kasl patterns for the BSS}
\label{fig:bss-kasl-spec-klone}
\end{figure}

In the BSS, if we assume at most one fault, then the sensor faults are
neither \dglobally\ nor \dlocally\ diagnosable since we are only able
to observe the difference in output of the sensors, and therefore we
can never be sure of which sensor is experiencing the
fault. Restricting the model to two faults, instead, makes it possible
to detect when both sensors are faulty, since the device stops
working.  The Battery Leak is \dlocally\ diagnosable but not
\dglobally\ diagnosable. This means that in general, we cannot detect
the battery leak, but there is at least one execution in which we
can. In particular, this is the execution in which the mode becomes
Secondary 2 when Battery 1 was charged, and we can see the battery
discharging, thus detecting the fault. Note that to detect this fault,
we need to recall the fact that previously the battery was charged,
and therefore a simple diagnoser without memory would not be able to
detect this fault.

% LocalWords:  epistemic diagnosability maximality LTL

\section{Validation and Verification of \texorpdfstring{\kasl}{ASLk} Specifications}
\label{sec-validation-verification}

Thanks to the formal characterization of \kasl, it is possible to
apply formal methods for the validation and verification of a set of
FDI requirements. In \emph{validation} we verify that the
requirements capture the interesting behaviors and exclude the
spurious ones, before proceeding with the design of the
diagnoser. In \emph{verification}, we check that a candidate
diagnoser fulfills a set of requirements.

\subsection{Validation}
\label{sec-fdi-validation}

Given a specification $\spec$ for our diagnoser, we want to make sure
that it captures the designer expectations.  Known techniques for
requirements validation (e.g.,\cite{CRST12}) include checking their
\emph{consistency}, and their \emph{realizability}, i.e., whether they
can be implemented on a given plant. Moreover, often
we want to show that there exists some condition under which the
alarm might be triggered (\emph{possibility}), and some other
conditions that require the alarm to be triggered (\emph{necessity}).

By construction, an \kasl specification is always \emph{consistent}, i.e.,
there are no internal contradictions. This is due to the fact that
alarm specifications do not interact with each other, and each alarm
specification can always be satisfied by a diagnosable
plant. Moreover, in Section~\ref{sec-fdi-synthesis}, we will prove
that we can always synthesize a diagnoser satisfying $\spec$, with the
only assumption that if $\spec$ contains some \dglobally\ diagnosable
alarm condition, then that condition is diagnosable in the plant.
Thus, the check for \emph{realizability} reduces to checking that the plant
is diagnosable for the \dglobally\ diagnosable conditions in $\spec$.
The diagnosability check can be performed via epistemic model-checking
(Section~\ref{sec-epistemic-diag}) or it can be reduced to an LTL
model-checking problem using the twin-plant construction~\cite{Cimatti2003}.

An alarm that is always (or never) triggered is not useful. Therefore,
we need to check under which conditions the alarm can and cannot be
triggered. Moreover, there might be some assumptions on the
environment of the diagnoser (including details on the plant) that
might have an impact on the the alarms. For example, if we have
a single fault assumption for our system, an alarm that
implicitly depends on the occurrence of two faults will never be
triggered. Similarly, our assumptions on the environment might provide
some link between the behavior of different components, or dynamics of
faults and thus characterize the relation between different alarms.

We consider a set of environmental assumptions $E$ expressed as LTL
properties. This set can be empty, or include detailed information on
the behavior of the environment and plant, since throughout the
different phases of the development process, we have access to better
versions of the plant model, and therefore the analysis can be
refined.

When checking \emph{possibility} we want that the alarms can be
eventually activated, but also that they are not always active. This
means that for a given alarm condition $\varphi \in \spec$, we are
interested in verifying that there is a trace $\sigma \in E$ and a
trace $\sigma' \in E$ s.t.\ $\sigma \models F \trigger{A_\varphi}$ and
$\sigma' \models F \lnot \trigger{A_\varphi}$. This can be done by
checking the unsatisfiability of $(E \land \varphi) \rightarrow G
\lnot \trigger{A_\varphi}$ and $(E \land \varphi) \rightarrow G
\trigger{A_\varphi} $.

Checking \emph{necessity} provides us a way to understand whether
there is some correlation between alarms. This, in turns, makes it
possible to simplify the model, or to guarantee some redundancy
requirement. To check whether $A_{\varphi'}$ is a more general alarm
than $A_\varphi$ (subsumption) we check whether $ (E \land \varphi
\land \varphi') \rightarrow G (\trigger{A_\varphi} \rightarrow
\trigger{A_{\varphi'}})$ is valid.  An example of subsumption of
alarms is given by the definition of maximality: any non-maximal alarm
subsumes its corresponding maximal version.
Finally, we can verify that two alarms are mutually exclusive by
checking the validity of $(E \land \varphi \land \varphi') \rightarrow
G \lnot (\trigger{A_\varphi} \land \trigger{A_{\varphi'}} )$.

To clarify the concepts presented in this section, we apply a
necessity check on our running example. In the Battery-Sensor, we have
two alarms specified on $PSU1$ (Figure~\ref{fig:bss-kasl-spec}):
$PSU1_{Exact_i}$ and $PSU1_{Bound}$. Let's take $i=C=2$, thus
obtaining:
\begin{itemize}
\item[-] $\kedel{PSU1_{Exact_2}}{\beta_{PSU1}}{2}{Trace}{True}$
\item[-] $\kbdel{PSU1_{Bound}}{\beta_{PSU1}}{2}{Trace}{True}$
\end{itemize}
\noindent
we want to show that $PSU1_{Exact_i}$ is more specific than (is
subsumed by) $PSU1_{Bound}$. This means that for any plant and
diagnoser, the following holds:
\[ D \otimes P \models (\varphi_{PSU1_{Exact_2}} \land \varphi'_{PSU1_{Bound}}) \rightarrow G (
\trigger{PSU1_{Exact_2}} \rightarrow \trigger{PSU1_{Bound}}) \]
\noindent By renaming with $PE = PSU1_{Exact_2}$ and $PB =
PSU_{Bound}$ (for brevity) and expanding the definitions of
$\varphi_{PSU1_{Exact_2}} \land \varphi'_{PSU1_{Bound}}$ we have that
\begin{align*}
D \otimes P \models
(G (\trigger{PE} \rightarrow Y^2 \beta) \land G ( \strong{KY^2\beta} \rightarrow \trigger{PE} ) \ \land \\
G (\trigger{PB} \rightarrow O^{\le 2} \beta) \land G (\strong{K O^{\le2} \beta} \rightarrow \trigger{PB}) ) \\
\rightarrow  G( \trigger{PE} \rightarrow \trigger{PB} )
\end{align*}
\noindent
We can apply Theorem~\ref{thm-alarm-implies-knowledge}, and therefore write:
\begin{align*}
D \otimes P \models (G (\trigger{PE} \rightarrow Y^2 \beta) \land G ( \strong{KY^2\beta} \rightarrow \trigger{PE} ) \ \land \\
G (\trigger{PB} \rightarrow O^{\le 2} \beta) \land G (\strong{K O^{\le2} \beta} \rightarrow \trigger{PB}) \ \land \\
 G( \trigger{PE} \rightarrow \strong{K Y^2\beta}) \land G( \trigger{PB} \rightarrow \strong{K O^{\le2}\beta})) \\
\rightarrow  G( \trigger{PE} \rightarrow \trigger{PB} )
\end{align*}
\noindent
To prove that the above formula is valid (and therefore it is satisfied
by any plant and diagnoser), we prove that its negation is
\emph{unsatisfiable}:
\begin{align*}
(G (\trigger{PE} \rightarrow Y^2 \beta) \land G ( \strong{KY^2\beta} \rightarrow \trigger{PE} ) \ \land \\
G (\trigger{PB} \rightarrow O^{\le 2} \beta) \land G (\strong{K O^{\le2} \beta} \rightarrow \trigger{PB}) \ \land \\
 G( \trigger{PE} \rightarrow \strong{K Y^2\beta}) \land G( \trigger{PB} \rightarrow \strong{K O^{\le2}\beta})) \\
\land \lnot G( \trigger{PE} \rightarrow \trigger{PB} )
\end{align*}
\noindent The first part of this formula is composed by conjuncts in
the form $G \psi$. This means that a counter examples is a trace for
which each state satisfies $\psi$. Moreover, we need one of
these states to satisfy $(PE \land \lnot \trigger{PB})$. Therefore, to
prove the unsatisfiable of the above formula, we can just prove that
no state exists that satisfies:
\begin{align*}
(\trigger{PE} \rightarrow Y^2 \beta) \land ( \strong{KY^2\beta} \rightarrow \trigger{PE} ) \ \land \\
(\trigger{PB} \rightarrow O^{\le 2} \beta) \land (\strong{K O^{\le2} \beta} \rightarrow \trigger{PB}) \ \land \\
(\trigger{PE} \rightarrow \strong{K Y^2\beta}) \land (\trigger{PB} \rightarrow \strong{K O^{\le2}\beta})) \\
\land \trigger{PE} \land \lnot \trigger{PB}
\end{align*}
\noindent We show this by a contradiction since:
\begin{align*}
& \cdots \land \trigger{PE} \land \lnot \trigger{PB} \\
\textbf{ \ObsPoint\ Def. } & \cdots \land \strong{\top} \land PE \land \lnot PB \\
\textbf{ Theorem~\ref{thm-alarm-implies-knowledge} on PE} & \cdots \land \strong{\top} \land PE \land \lnot PB \land K YY \beta  \\
\textbf{ Maximality of PB } & \cdots  \land \strong{\top} \land PE \land \lnot PB \land K YY \beta \land \lnot K O^{\le 2} \beta \\
\textbf{ $\dagger$Def. of $\lnot K$ } & \cdots  \land \strong{\top} \land PE \land \lnot PB \land K YY \beta \land \lnot O^{\le 2} \beta \\
\textbf{ Def. of $O^{\le n}$ } & \cdots \land \strong{\top} \land PE \land \lnot PB \land K YY \beta \land \lnot (\beta \lor Y \beta \lor YY\beta) \\
\textbf{ K Axiom ($K \phi \rightarrow \phi$) } & \cdots \land \strong{\top} \land PE \land \lnot PB \land YY \beta \land \lnot \beta \land \lnot Y \beta \land \lnot YY\beta\\
\end{align*}
\noindent
Thus reaching a contradiction between $YY\beta$ and $\lnot
YY\beta$. In the step marked with $\dagger$ we need to show that two
observationally equivalent traces exists s.t.\ one satisfies $O^{\le
  2} \beta$ and the other $\lnot O^{\le 2} \beta$; therefore, we only
need to show that one of the two (namely $\lnot O^{\le 2} \beta$) does
not exist.

\subsection{Verification}
\label{sec-fdi-verification}
The verification of a system w.r.t.\ a specification can be performed
via model-checking techniques using the semantics of the alarm
conditions:
\begin{defi}
Let $D$ be a diagnoser for alarms $\alarms$ and plant $P$. We say that
\emph{$D$ satisfies a set $\spec$ of \kasl specifications} iff for each
$\varphi$ in $\alarms_\props$ there exists an alarm $A_\varphi \in
\alarms$ and $D \otimes P \models \varphi$.
\end{defi}
\noindent
To perform this verification steps, we need in general a model checker
for \klone\ with asynchronous/synchronous perfect recall such as
MCK~\cite{mck}. However, if the specification falls in the pure LTL
fragment (ASL) we can verify it with an LTL model-checker such as
nuXmv~\cite{nuxmv} thus benefiting from the efficiency of
the tools in this area.

Moreover, a diagnoser is required to be deterministic. This is
important, on one hand, for implementability, on the other hand, to
ensure that the composition of the plant with the diagnoser does not
reduce the behaviors of the plant. In order to verify that a given
diagnoser $D=\mktuple{V,E,I,\mathcal{T}}$ is deterministic, we check
the following conditions:
\begin{itemize}
\item
$I$ must be satisfiable,
\item
$I\wedge I[V_c/V] \rightarrow V=V_c$ must be valid,
\item
for all $e\in E$, $\forall V\exists V'.\mathcal{T}(e)$ must be valid (note that
this corresponds to the validity of the pre-image of $\top$),
\item
for all $e\in E$, $\mathcal{T}(e) \wedge \mathcal{T}(e)[V_c/V'] \rightarrow V'=V_c$ must be valid.
\end{itemize}
Therefore, we can solve the problem with a finite set of
satisfiability checks and pre-image computations.

\section{Synthesis of a Diagnoser from an \texorpdfstring{\kasl}{ASLk} Specification}
\label{sec-fdi-synthesis}

In this section, we discuss how to synthesize a diagnoser that
satisfies a given specification $\spec$. We considers the most
expressive case of \kasl (maximal/\dlocally\ diagnosable), which also
satisfies all the other cases.

The idea is to generate an automaton that encodes the set of possible
states in which the plant could be after each observations. The result
is achieved by generating the power-set of the states of the plant,
also called \emph{belief states}, and
defining a suitable transition relation among the elements of this
set, only taking into account observable information. Each belief
state of the automaton is then annotated with the alarms that are
satisfied in all the states of the belief state. The resulting
automaton is the Diagnoser.

The approach resembles the constructions by Sampath~\cite{Sampath96}
and Schumann~\cite{Schumann2004}, with the following main differences.
First, we consider LTL Past expression as diagnosis condition, and not
only fault events as done in previous works. Second, instead of
providing a set of possible diagnoses, we provide alarms. In order to
raise the alarm, we need to be certain that the alarm condition is
satisfied for all possible diagnoses. This gives raise to a 3-valued
alarm system: we \emph{know} that the fault occurred; \emph{know} that
the fault did \emph{not} occur; or we are \emph{uncertain}. Moreover,
the approach works for the asynchronous case. Although the use of a
power-set construction in the setting of temporal epistemic logic is
not novel (e.g.~\cite{dima2009revisiting} for synchronous CTLK
model-checking), the main contribution of this section is to show the
formal properties of the diagnoser, and in particular that it
satisfies the specification. In a way, this algorithm is a strong
indicator of a deep connection between the topics of temporal
epistemic logic reasoning and FDI design.

\subsection{Synthesis algorithm}

Given a partially observable plant $P=\posystem[P]$, let $S$ be the
set of states of $P$. The \emph{belief automaton} is defined as
$\mathcal{B}(P) = \mktuple{B,E,b_0,R}$ where $B=2^{S}$, $E=E^P_o$,
$b_0\in B$ and $R : (B \times E) \rightarrow B$.  $B$ represents the
set of sets of states, also called \emph{belief states}. Given a
belief state $b$, we use $b^*$ to represent the set of states that are
reachable from $b$ by only using events in $E^P\setminus E^P_o$ (non
observable events), and call it the $u$-transitive closure. Formally,
$b^*$ is the least set s.t.\ $b\subseteq b^*$ and if there exist $e \in
E^P \setminus E^P_o$ and $s' \in b^*$ such that $\mktuple{s',s} \in
\mathcal{T}^P(e)$ then $s\in b^*$.
$b_0$ is the initial belief state and contains the states that satisfy
the initial condition $I^P$ (i.e., $b_0=\{s\mid s \models I^P\}$).

Given a belief state $b$ and an observable event $e \in E^P_o$, we
define the successor belief state $b'$ as:
\[R(b, e) = b' = \{ s' \mid \exists s \in b^* .\ \mktuple{s,s'} \models T^P(e)\} \]

\noindent that is the set of states that are compatible with the
observable event \emph{e} in a state of the $u$-transitive closure of
$b$. Intuitively, we first compute the $u$-transitive closure of $b$
to account for all non-observable transitions, and then we consider
all the different states that can be reached from $b^*$ with an
occurrence of the event $e$.

The diagnoser is obtained by annotating each state of the belief
automaton with the corresponding alarms. We annotate with $A_\varphi$
all the states $b$ that satisfy the temporal property
$\tau(\varphi)$. As explained later on, any temporal $\tau(\varphi)$
can be handled by introducing suitable propositional formulas.
Therefore we consider the simplest case in which $\tau(\varphi)$ is a
propositional formula and formally say that the annotation $a_b$ of the belief
state $b$ is the assignment to $A_\varphi$ such that $a_b(A_\varphi)$
is true iff for all $s \in b$, $s \models \tau(\varphi)$.
We perform the same annotation for $A_{\lnot \varphi}$.
The diagnoser obtained by this algorithm
induces three alarms, related to the knowledge of the diagnoser. In
particular, the diagnoser can be sure that a condition occurred
($A_\varphi$) can be sure that a condition did not occur
($A_{\lnot\varphi}$) or can be uncertain on whether the condition
occurred ($\lnot A_{\varphi} \land \lnot A_{\lnot \varphi}$) -- notice
that, by construction, it is not possible for both $A_{\varphi}$ and
$A_{\lnot \varphi}$ to be true at the same time. In this way, at any
point in time we are able to understand whether we are on
a trace that is not diagnosable (and thus there is uncertainty) or
whether the diagnoser knows that the condition did not occur. This can
thus provide additional insight on the behavior of the system.

\begin{figure}
  \begin{algorithmic}
  \Function{belief\_automaton}{$I$, $T$, $E$, $E_o$}
     \State $visited \gets \{\}$
     \State $edges \gets \{\}$

     \State $stack \gets [I]$
  \While{not $stack.is\_empty()$}
    \State $b \gets stack.pop()$
    \State $b^* \gets u\_trans\_closure(b, T, E)$
    \ForAll{$o \in get\_observable\_events(b^*, T, Eo)$}
      \State $target\_belief \gets reachable\_w\_obs(b^*, o, T)$
      \State $edges.add((b, o, target\_belief))$
      \If{$target\_belief \not \in visited$}
        \State $visited.add(target\_belief)$
        \State $stack.push(target\_belief)$
      \EndIf
    \EndFor
  \EndWhile
  \State \textbf{return} $Automaton(visited, edges)$
  \EndFunction
\end{algorithmic}
\caption{Pseudo-code of the Belief Automaton construction phase}
\label{fig:synth-pseudo-code}
\end{figure}

Figure~\ref{fig:synth-pseudo-code} provides a
pseudo-code of the main function of the synthesis task: the
construction of the belief automaton. Starting from the set of initial
states, we perform an explicit visit until we have explored all belief
states. For each belief state we first compute its $u$-transitive
closure ($u\_trans\_closure$) w.r.t.\ the non-observable events $E$,
obtaining $b^*$. We then compute the possible observable events
available from $b^*$, and iterate over each event $o_i$ obtaining
the set of states $target\_belief$ such that $T(b\_star, o_i,
target\_belief)$ is satisfied ($reachable\_w\_obs$). We can now add a
transition to our automaton linking the belief state $b$ to the
belief state $target\_belief$ through the event $o_i$. Once we have completed this phase,
we have an automaton with labeled transitions. The automaton resulting
from this function can then be annotated by visiting each state and
testing whether the state entails (or not) the alarm specification.

\subsection{Running Example}
We show the first step of the algorithm on a simplified version of
the battery component of our running
example (Figure~\ref{fig:BatterySensor-Battery-LTS}). We ignore the
events related to threshold passing of the battery (\emph{Mid}, \emph{Low},
\emph{High}) and only consider the observable event \emph{Off},
signaled when the charge reaches zero, and the ones due to mode
changes. To keep the representation
compact, we indicate each state with three symbols. For example, we
use $(NPC)$ to indicate the state ``Nominal, Primary, Charging'' and
$(NP\overline{C})$ to indicate the state ``Nominal, Primary, Not
Charging''. Similarly we use \emph{F}, \emph{O}, and \emph{D} to
indicate Faulty, Offline and Double. We recall that in the original
model, the mode transitions are observable but all other transitions
are not.

In the first step (Figure~\ref{fig:synth-step1}), we take the set of
initial states. This is the set of states $(NPC)$ for any value of the
$charge \in [0,C]$. The $u$-transitive closure needs to take into
account all non-observable transitions. Therefore, we need to consider
going from Nominal to Faulty, from Charging to Not Charging, and their
combination.
\begin{figure}[ht]
  \resizebox{!}{0.15\textwidth}{
\begin{tikzpicture}[->,>=stealth',shorten >=1pt,
    auto,node distance=3.5cm, semithick,
   ]
  \tikzstyle{every state}=[draw=black]
  \tikzstyle{observableT}=[color=black]
  \node[initial,state] (NPC) {$(NPC)$};
  \node[text width=1.5cm, state] (B1) [right of=NPC]
       {$(NPC)$\\$(NP\overline{C})$\\
        $(FPC)$\\$(FP\overline{C})$};

  \path (NPC) edge node {$u$-trans} (B1);
\end{tikzpicture}}
\caption{Expanding the initial belief state of the battery LTS.}
\label{fig:synth-step1}
\end{figure}

These are all the states that are reachable before an observable event
can occur. We now take each observable event and compute the set of
states that are reachable with one of the observable events
(Figure~\ref{fig:synth-step2}): the battery being discharge
(\emph{Off}), and the change of mode (\emph{Offline}, \emph{Double}).
\begin{figure}[ht]
\begin{minipage}{0.5\textwidth}
\resizebox{!}{0.55\textwidth}{
\begin{tikzpicture}[->,>=stealth',shorten >=1pt,
    auto,node distance=3.5cm, semithick,
   ]
  \tikzstyle{every state}=[draw=black]
  \tikzstyle{observableT}=[color=black]
  \tikzstyle{textlimit}=[text width=1.5cm,align=center]
  \node[initial,state] (NPC) {$(NPC)$};
  \node[textlimit, state] (B1) [right of=NPC]
       {$(NPC)$\\$(NP\overline{C})$\\
        $(FPC)$\\$(FP\overline{C})$};
  \node[textlimit, state] (B2) [below of=B1]
       {$(NOC)$\\$(NO\overline{C})$\\
        $(FOC)$\\$(FO\overline{C})$};
  \node[textlimit, state] (B3) [right of=B2]
       {$(NDC)$\\$(ND\overline{C})$\\
        $(FDC)$\\$(FD\overline{C})$};
  \node[textlimit, state] (B4) [right of=B1]
       {$(NP\overline{C})$\\
        $(FPC)$\\$(FP\overline{C})$};

  \path (NPC) edge node {$u$-trans} (B1);
  \path (B1)
  edge node [left] {Offline} (B2)
  edge node {Double} (B3)
  edge node {Off} (B4);
\end{tikzpicture}}
\end{minipage}\begin{minipage}{0.5\textwidth}
\resizebox{!}{0.55\textwidth}{
\begin{tikzpicture}[->,>=stealth',shorten >=1pt,
    auto,node distance=3.5cm, semithick,
   ]
  \tikzstyle{every state}=[draw=black]
  \tikzstyle{observableT}=[color=black]
  \tikzstyle{textlimit}=[text width=1.5cm,align=center]
  \node[initial,state] (NPC) {$(NPC)$};
  \node (B1) [right of=NPC] {};
  \node[textlimit, state] (B2) [below of=B1]
       {$(NOC)$\\$(NO\overline{C})$\\
        $(FOC)$\\$(FO\overline{C})$};
  \node[textlimit, state] (B3) [right of=B2]
       {$(NDC)$\\$(ND\overline{C})$\\
        $(FDC)$\\$(FD\overline{C})$};
  \node[textlimit, state] (B4) [right of=B1]
       {$(NP\overline{C})$\\
        $(FPC)$\\$(FP\overline{C})$};

  \path (NPC)
  edge node [left] {Offline} (B2)
  edge node {Double} (B3)
  edge node {Off} (B4);
\end{tikzpicture}}
\end{minipage}
\caption{Expanding the belief state via observable transitions}
\label{fig:synth-step2}
\end{figure}
\noindent Note that one of the belief states is smaller than the
others. This is due to the fact that in our model, the discharging of
the battery cannot occur if the battery is nominal, charging and in
primary mode ($NPC$). Thus, the fact that we receive the $Off$ event
allows us to exclude that state. The state obtained by computing the
transitive closure is not part of our final automaton, and is provided
in the figure only to simplify the understanding.

We repeat these two steps until all belief states have been
explored. We then proceed to the labeling phase, in which we label
each state with the corresponding alarm. For example, by considering the
alarms $\edel{A_{NC}}{Nominal \land Charging}{0}$ and
$\edel{A_{N}}{Nominal}{0}$, we obtain the diagnoser partially
represented in Figure~\ref{fig:synth-step3}. Notice how, in the
initial state we can raise the alarm $A_{NC}$, and this alarm can only
be changed by an observable transition.

\begin{figure}[t]
  \resizebox{!}{0.32\textwidth}{
\begin{tikzpicture}[->,>=stealth',shorten >=1pt,
    auto,node distance=3.5cm, semithick,
   ]
  \tikzstyle{every state}=[draw=black]
  \tikzstyle{observableT}=[color=black]
  \tikzstyle{textlimit}=[text width=1.5cm,align=center]
  \node[initial, textlimit, state] (NPC)
       {$A_{NC}$,\\$A_{N}$\\
        $\lnot A_{\lnot NC}$\\$\lnot A_{\lnot N}$};
  \node (B1) [right of=NPC] {};
  \node[textlimit, state] (B2) [below of=B1]
       {$\lnot A_{NC}$\\$\lnot A_{N}$\\
        $\lnot A_{\lnot NC}$\\ $\lnot A_{\lnot N}$};
  \node[textlimit, state] (B3) [right of=B2]
       {$\lnot A_{NC}$\\$\lnot A_{N}$\\
        $\lnot A_{\lnot NC}$\\$\lnot A_{\lnot N}$};
  \node[textlimit, state] (B4) [right of=B1]
       {$\lnot A_{NC}$\\$\lnot A_{N}$\\
        $A_{\lnot NC}$\\$\lnot A_{\lnot N}$};
   \node (dots1) [left of=B2] {$\cdots$};
   \node (dots2) [right of=B3] {$\cdots$};
   \node (dots3) [right of=B4] {$\cdots$};

  \path (NPC)
  edge node [left] {Offline} (B2)
  edge node {Double} (B3)
  edge node {Off} (B4);
  \path (B2)
  edge node {} (B3)
  edge [bend right] node {} (dots1);

  \path (B3)
  edge node {} (B2)
  edge [bend left] node {} (dots2);

  \path(B4)
  edge node {} (B3)
  edge node {} (B2)
  edge [bend left] node {} (dots3);
\end{tikzpicture}}
\caption{Annotation of the belief states}
\label{fig:synth-step3}
\end{figure}

\subsection{Formal Properties of the Synthesized diagnoser}
\label{sec:synth-formal-prop}
We now show that the generated transition system is a diagnoser and
that it is correct, complete and maximal.  Lets assume that $\varphi$
is an exact delay specification, with delay zero. Any other alarm
conditions can be reduced to this case. We build a new plant $P'$ by
adding a monitor variable $\overline{\tau}$ to $P$ s.t., $P' = P
\times (G (\tau(\varphi) \leftrightarrow \overline{\tau}))$, where we
abuse notation to indicate the synchronous composition of the plant
with an automaton that encodes the monitor variable. By rewriting the
alarm condition as $\varphi' = \edel{A_\varphi}{\overline{\tau}}{0}$,
we obtain that $D \otimes P \models \varphi$ iff $D \otimes P' \models
\varphi'$. Thus, it is sufficient to show the following results
only for the zero delay case.
We define $D_\varphi$ as the diagnoser for $\varphi$.
$D_\varphi=\system[D_\varphi]$ is a symbolic representation of
$\mathcal{B}(P)$ with $A_\varphi \subseteq V^{D_\varphi}$,
$E^{D_\varphi}_o=E^P_o$ and such that every state $b$ of $D_\varphi$
represents a state in $B$ (with abuse of notation we do not
distinguish between the two since the assignment to $A_\varphi$ is
determined by $b$).

\begin{thm}\label{th-compatibility}
$D_\varphi$ is deterministic.
\end{thm}

\begin{proof}
The result follows directly from the definition of the belief
automaton, which is deterministic (one initial state and one
successor). Note that the assignment to $A_\varphi$ is not relevant since
it is determined by the belief state.
\end{proof}

\begin{lem}\label{lem-observations}
For every reachable state $b \times s$ of $ D_\varphi \otimes P$, for
every trace $\sigma$ reaching $b \times s$, for every state $s' \in
b$, there exists a trace $\sigma'$ reaching $b \times s'$ with
$\Obs(\sigma)=\Obs(\sigma')$.
\end{lem}

\begin{proof}
By induction on $\sigma$. All traces are observationally equivalent in
the initial state.  Let $\mktuple{b_1\times s_1, e, b \times
  s}$ be the last transition of $\sigma$ and let $\sigma_1$ be the prefix
of $\sigma$ without this last transition. If $e \in E \setminus E_o$
then $\Obs(\sigma)=\Obs(\sigma_1)$. Otherwise, for every state $s'\in
b$ there exists a transition $\mktuple{s'_1,e,s'}$ such that $s'_1\in
b_1^*$. By inductive hypothesis there exists a trace $\sigma'_1$ reaching
$b_1 \times s'_1$ such that $obs(\sigma_1)=obs(\sigma'_1)$. Therefore
the concatenation of $\sigma'_1$ with the transition $\mktuple{b_1
  \times s'_1, e, b \times s'}$ results in a trace $\sigma'$ reaching
$b \times s'$ such that $obs(\sigma)=obs(\sigma')$.
\end{proof}

\begin{thm}[Maximality]\label{th-maximality}
$D_\varphi \otimes P \models G(\strong{K\tau(\varphi)}\rightarrow \trigger{A_\varphi})$.
\end{thm}

\begin{proof}
Consider a trace $\sigma$ and $i\geq 0$. If $\sigma,i\models
\strong{K \tau(\varphi)}$, then for all traces $\sigma'$ and points $j$
s.t.\ $ObsEq((\sigma, i), (\sigma', j))$, $\sigma', j\models \tau(\varphi)$.  By
Lemma~\ref{lem-observations}, all states $s\in \sigma[i]$ there exists
a trace $\sigma'$ with $\Obs(\sigma)=\Obs(\sigma')$, and therefore
$s\models \tau(\varphi)$ so that $\sigma[i]\models \trigger{A_\varphi}$.
\end{proof}

\begin{lem}\label{lem-beliefinclusion}
Given a trace $\sigma$ of $D_\varphi\otimes P$. Let $\sigma[i]=b\times
s$. If $i$ is an observation point, then $s\in b$.
\end{lem}

\begin{proof}
By assumption, $i$ is the $n$-th observation point of $\sigma$ for
some $n$. We prove the lemma by induction on $n$.

Consider the case $n=1$.
If $\sigma[0]=b_0\times s_0$, by construction of $D_\varphi$, $s_0\in
b_0$. Let $\sigma[i-1]=b'\times s'$ and let $e$ be the $i$-th
(observable) event of $\sigma$. If $i$ is the first observation point
of $\sigma$, it means that $b'=b_0$ and $s'\in b_0^*$. Moreover,
$\mktuple{s',s}\in T(e)$ and therefore $s\in b$.

Consider the case $n>1$. Let $j$ be the $n-1$ observation point,
$\sigma[j]=b_j\times s_j$, $\sigma[i-1]=b'\times s'$ and let $e$ be
the $i$-th (observable) event of $\sigma$. Similarly to the previous
case, $b'=b_j$ and $s'\in b_j^*$. Moreover $\mktuple{s',s}\in T(e)$ and
therefore $s\in b$.
\end{proof}

\begin{thm}[Correctness]\label{th-correctness}
$D_\varphi \otimes P\models G(\trigger{A_\varphi} \rightarrow \tau(\varphi))$.
\end{thm}

\begin{proof}
Consider a trace $\sigma$ and $i\geq 0$. Suppose $\sigma, i\models
\trigger{A_\varphi}$ and let $\sigma_{D_\varphi}$ and $\sigma_P$ be
respectively the left and right component of $\sigma$. Then, for all
$s \in \sigma_{D_\varphi}[i]$, $s\models \tau(\varphi)$. Since $i$ is
an observation point, by Lemma~\ref{lem-beliefinclusion},
$\sigma_P[i]\in\sigma_{A_\varphi}[i]$.
%\todo{R1: The connection between these two proofs is unclear})
We can conclude that $\sigma[i]\models \tau(\varphi)$.
\end{proof}

\begin{thm}[Completeness]\label{th-ccm}
If $\varphi$ is an alarm condition required to be trace diagnosable, then $D_\varphi$ is
complete. If $\varphi$ is a system diagnosable condition and $\varphi$
is diagnosable in $P$, then $D_\varphi$ is complete.
\end{thm}

\begin{proof}
Since $D_\varphi$ is maximal and correct (Theorems~\ref{th-maximality}
and \ref{th-correctness}), we can apply
Theorem~\ref{thm-max-local-for-completeness} (if $\varphi$ is
\dlocally\ diagnosable) or Theorem~\ref{thm-max-global-for-completeness}
(if it is \dglobally\ diagnosable) to obtain completeness.
\end{proof}

\section{Industrial Experience}
\label{sec-industrial}

The methods described in this paper have been motivated by AUTOGEF,
a project~\cite{AUTOGEF-ITT,AUTOGEF-WEBSITE,AUTOGEF-DASIA} funded by
the European Space Agency.  The main goal of the project was the
definition of a set of requirements for an on-board Fault Detection,
Identification and Recovery (FDIR) component and its synthesis. The
problem was cast in the frame of discrete event systems, communicating
asynchronously, and tackled by synthesizing the Fault Detection (FDI)
and Fault Recovery (FR) components separately -- with the idea that
the FDI provides sufficient diagnosis information for the FR to act
on.

A similar problem was further investigated in FAME, another ESA-funded
project~\cite{FAME-ITT,FAME-WEBSITE,FAME-DASIA,IMBSA-regular,IMBSA-tool}. In
the context of FAME, we addressed the problem of synthesis of FDI and
FR components for continuous time systems, with synchronous
communication -- in particular the diagnoser communicates with the
plant by sampling the values of the sensors at periodic time
intervals.  In both cases, AUTOGEF and FAME, we addressed the problem
of \textsc{FiniteDel} diagnosis, which was of interest from an
industrial perspective.

Within AUTOGEF, the design approach initially was evaluated using scalable
benchmark examples. Then, Thales Alenia Space evaluated AUTOGEF on an
industrial case study based on the EXOMARS Trace Gas Orbiter. This
case-study is a significant application of the approach described in
this paper, since it covers all the phases of the FDIR development
process.  The (nominal and faulty) behavior of the system was modeled
using a formal language. A table-based and pattern-based approach was
adopted to describe the mission phases/modes and the observability
characteristics of the system.  The specification of FDIR requirements
by means of patterns greatly simplified the accessibility of the tool
to engineers that were not experts in formal methods.  Alarms were
specified in the case of finite delay, under the assumption of
\dlocal\ diagnosability and maximality of the diagnoser.  Different
faults and alarms were associated with specific mission phases/modes
and configurations of the system, which enabled generation of specific
alarms (and recoveries) for each configuration. The specification was
validated, by performing diagnosability analysis on the system
model. The synthesis routines were run on a system composed of 11
components, with 10 faults in total, and overall 90 bits of variables,
and generated an FDI component with 754 states. Finally, the correctness of the diagnoser was verified by
using model-checking routines.
Synthesis and verification capabilities have been implemented on top
of the nuXmv model checker.
We remark that the ability to define \dlocally\ diagnosable alarms was
crucial for the synthesis of the diagnoser, since most of the modeled
faults were not \dglobally\ diagnosable.

A similar approach was undertaken in FAME. The industrial evaluation was carried
out on a further elaboration of the Trace Gas Orbiter case study,
adapted to take into account timings of fault propagation. The
specification of the FDIR requirements and the verification,
validation and synthesis process were done in a similar way. As a
difference with AUTOGEF, the synthesis of FDI in FAME was aided by the
specification of a fault propagation model, in the form of a Timed
Failure Propagation Graph (TFPG)~\cite{IMBSA-regular,TFPGAAAI15}.  The case study
investigated fault management related to the feared event `loss of the
spacecraft attitude'.  A total of 3 faults, instantiated for two
(redundant) instances of the Inertial Management Unit (IMU) component
were considered. The synthesis of FDI produced an FDI component with
2413 states.

Successful completion of both projects, and positive evaluations from
the industrial partner and ESA, suggest that a significant first step towards a
formal model-based design process for FDIR was achieved.

\section{Related Work}
\label{sec-related}

\subsection{From Synchronous to Asynchronous FDI}
\label{sec-sync-into-async}

This work is closely related
to~\cite{DBLP:conf/tacas/BozzanoCGT14}. The key difference is that we
extended the approach to include the asynchronous composition of the
plant with the diagnoser.
This extension is useful in practice, since many real-life systems
as well as many high-level modeling languages
adopt an asynchronous, event-based view. In the
synchronous case system and diagnoser share the same time scale, and
the diagnoser takes a step every time the system does. In the
asynchronous setting, on the other hand, the diagnoser takes a step
only when the system exhibits an observable behavior, (i.e., an
observable event).

Although this could be seen as a minor difference, it poses nontrivial
problems.
First of all, since the diagnoser cannot update the value of the
alarms at every point in time,
we need to restrict the definition of Correctness and Completeness to
the occurrence of a synchronization, in which the diagnoser can update
the alarms, by introducing observation points and using the \emph{observed}
version $\trigger{A}$ of $A$.
Similarly, since the diagnoser can update its knowledge of the plant
only during synchronizations, also the epistemic operator is
considered in the observation points. Therefore, we define $K$ as
usual, but then introduce a stronger version $\strong{K}$ that is the
basis for most of our definitions.

The synthesis algorithm also needs to take into account multiple
transitions from the plant that are executed without
synchronization. This is done by introducing the
$u$-transitive closure of the belief states.

Finally, to keep the formalism simple, we modeled the
observability of state variables as observable events.
This is mainly due to the fact that a change
in observable state variables requires the introduction of a new
synchronization event between the plant and the diagnoser in order to
allow the diagnoser to update its knowledge.
This idea is consistent with the approach defined in
\cite{Sampath96}. Also, in other works on knowledge in an asynchronous
setting~(e.g., \cite{Mey07}), the fact that the observer sees every
observable state changes implicitly assumes that the observable state
change triggers a synchronization.  Note that this is somehow
different from asynchronous systems with shared variables, where a
process can see the change of the shared variable only when/if
scheduled.

We notice that the synchronous case can be embedded in the
asynchronous one. In fact, according to Def.~\ref{def:sync-product}, a
synchronous product is obtained by making all events of the plant
observable: $E^P_o = E^P$. This implies that all points are
observation points. Therefore, $\trigger{A} = A$, and the
restriction of $K$ to observation points has no effect.  Also the
$u$-transitive closure has no effect, and we see that $b^*=b$.

\subsection{FDI Specification}

In order to formally verify the effectiveness of an FDI component as
part of an overall fault-management strategy, both a formal model of
the FDI component (e.g., as an automaton) and of its expected behavior
(requirements) is required. Contrary to works related to diagnosis
compilation, we are also interested in verifying that an FDI satisfies
a given specification. This has tremendous value when we consider the
problem of checking whether an existing system (that is familiar to
the system designer) satisfies the specification and thus is
functionally equivalent to an automatically synthesized one (that
could be complex and hard to understand).

Previous works on formal FDI development have considered the
specification and synthesis in isolation. Our approach
differs with the state of the art because we provide a comprehensive
view on the problem. Due to the lack of specification formalism for
diagnosers, the problem of verifying their correctness, completeness
and maximality was, to the best of our knowledge, unexplored.

Concerning specification and synthesis, \cite{Jiang2001} is close to
our work. The authors present a way to specify the diagnoser using LTL
properties, and present a synthesis algorithm for this
specification. However, problems such as maximality and
\dlocal\ diagnosability are not taken into account.
Another remarkable difference is that \cite{Jiang2001} considers
diagnosis conditions with future operators. This enables the
definition of alarms that predict the occurrence of an event
(i.e. prognosis), that is currently not captured in our work.

\subsection{Diagnosability}

In many practical situations it is not possible to require
\dglobal\ diagnosability, due, for example, to critical pairs that
exists only in a particular configuration of the system. We introduce
the concept of \dlocal\ diagnosability, that is a distinguishing
feature of our approach, and overcomes a strong limitation in the
current state-of-the-art.

The idea of using epistemic properties to analyze the diagnosability
of a system had been already proposed in~\cite{Ezekiel2011} and
\cite{Huang2013}. Notably, the latter extends the problem to a
probabilistic setting, and draws a link with the classical definition
of diagnosability, introducing the idea of L-diagnosability (that is
equivalent to our finite-delay diagnosability). Our approach extends
these works by considering other types of delay and the problem of
\dlocal\ diagnosability. Moreover, we do not focus only on the
diagnosability problem, but also provide a way of specifying the
diagnoser and characterize its completeness in terms of epistemic
temporal logic.

We extend the results on diagnosability checking from
\cite{Cimatti2003} in order to provide an alternative way of checking
diagnosability and redefine the concept of diagnosability at the trace
level.

\subsection{Runtime Verification}

The
main difference between diagnosis and runtime verification is the
partial observability of the plant. Works on runtime verification
assume~\cite{havelund2004efficient} that the properties to be
verified are expressed over observable variables of the system. In
diagnosis, instead, we define the properties over non-observable parts
of the system and then ask whether it is possible to infer them
by looking at the observable part of the system. Therefore,
while some approaches for runtime verification do not need a model of
the system (i.e., black-box approach), in diagnosis we need to have
some information about the behavior of the system. Finally,
in~\cite{bauer2011runtime} the authors propose the use of a
three-valued LTL variant to define whether a trace satisfies a
property, does not satisfy it or whether there is not enough information
to come to a conclusion. This might resemble the approach presented in
Section~\ref{sec-fdi-synthesis} by our synthesis algorithm. However,
the difference is substantial. Every time our diagnoser is uncertain,
it means that there are two traces $\sigma_1$ and $\sigma_2$ that are
observationally equivalent, but one satisfies the property and the
other does not. However, if we could have an oracle that would tell us
whether the system is in $\sigma_1$ or in $\sigma_2$, we could state
(without uncertainty) whether the property is satisfied or
not. In~\cite{bauer2011runtime} instead, the inconclusiveness of the
monitor is intrinsic in the fact that the given trace does neither
satisfy nor violate the property.

\section{Conclusions and Future Work}
\label{sec-conclusions}

This paper presents a formal approach for the design of FDI
components, that covers many practically-relevant issues such as
delays, non-diagnosability and maximality. The design is based on a
formal semantics provided by temporal epistemic logic and can be used
both in a synchronous and asynchronous setting.  We cover the
specification, validation, verification and synthesis steps of the FDI
design, and discuss the applicability of the approach on a case-study
from aerospace. To the best of our knowledge, this is the first work
that provides a formal and unified view to all the phases of FDI
design.

In the future, we plan to explore the following research directions.
First, we will extend FDI to deal with infinite-state
systems. Secondly, we will experiment with different assumptions on
the memory requirements for the diagnoser, i.e., relax the perfect
recall assumption.

Another interesting line of research is the development of optimized
reasoning techniques for temporal epistemic logic. The idea is
to consider the fragment that we are using, both for verification and validation, and
to evaluate and improve the scalability of the synthesis algorithms.

Finally, we will work on integrating the FDI component with the recovery
procedures.

\bibliographystyle{alpha}
\bibliography{refs}

\end{document}